%% file: arXiv.tex
\documentclass[conference,a4paper]{IEEEtran}

\usepackage[utf8]{inputenc} 
\usepackage[T1]{fontenc}
\usepackage{url}              %
\usepackage{cite}             %

\usepackage[cmex10]{amsmath}  %
\usepackage{mleftright}       %
\mleftright                   %

\usepackage{graphicx}         %
\usepackage{booktabs}         %

\usepackage{amsmath,graphicx}

\usepackage[ruled,vlined]{algorithm}
\usepackage{colortbl}
\usepackage{booktabs}
\usepackage{multirow}
\usepackage{caption}
\usepackage{float}
\usepackage{multicol}
\usepackage[normalem]{ulem}
\usepackage{algpseudocode}
\algnewcommand{\Initialize}{%
  \State \textbf{Initialize:}
}
\usepackage{mathtools}
\usepackage{graphicx}
\usepackage{url}

\usepackage[colorlinks=true, allcolors=blue]{hyperref}
\usepackage[
    shortcuts,       %
    acronym,         %
    section=section, %
]{glossaries}
\usepackage{tikz}
\usetikzlibrary{backgrounds,calc,shadings,shapes.arrows,shapes.symbols,shadows,fit,positioning,topaths,hobby}
\usepackage{float}
\usepackage{amssymb}
\usepackage{amsmath}
\usepackage{amsthm}
\usepackage{xspace}
\usepackage[capitalize]{cleveref}
\usepackage{subfig}
\usepackage[inline]{enumitem}
\usepackage{cite}
\usepackage{siunitx}
\makeglossaries

\input{acronyms.tex}

\input{commands.tex}

\newtheorem{theorem}{Theorem}

\newtheorem{definition}{Definition}

\newtheorem{corollary}[theorem]{Corollary}

\newtheorem{proposition}[theorem]{Proposition}

\IEEEoverridecommandlockouts

\begin{document}

\title{Fast and Straggler-Tolerant Distributed SGD\\ with Reduced Computation Load\vspace{-.2cm}} 

\author{%
  \IEEEauthorblockN{Anonymous Authors}
}

\author{%
  \IEEEauthorblockN{\textbf{Maximilian Egger}\IEEEauthorrefmark{1},
                    \textbf{Serge Kas Hanna}\IEEEauthorrefmark{2},
                    and \textbf{Rawad Bitar}\IEEEauthorrefmark{1}}
                    
  \IEEEauthorblockA{\IEEEauthorrefmark{1}%
                    Technical University of Munich, Munich,
                    Germany,
                    \{maximilian.egger, rawad.bitar\}@tum.de}
                    
  \IEEEauthorblockA{\IEEEauthorrefmark{2}%
                    Department of Mathematics and Systems Analysis, School of Science, Aalto University,
                    Finland,
                    serge.kashanna@aalto.fi}
                     
\thanks{This project has received funding from the German Research Foundation (DFG) under Grant Agreement Nos. BI 2492/1-1 and WA 3907/7-1.} \vspace{-.5cm}
}

\maketitle

\begin{abstract}
In distributed machine learning, a central node outsources computationally expensive calculations to external worker nodes. The properties of optimization procedures like stochastic gradient descent (SGD) can be leveraged to mitigate the effect of unresponsive or slow workers called stragglers, that otherwise degrade the benefit of outsourcing the computation. This can be done by only waiting for a subset of the workers to finish their computation at each iteration of the algorithm. Previous works proposed to adapt the number of workers to wait for as the algorithm evolves to optimize the speed of convergence. In contrast, we model the communication and computation times using independent random variables. Considering this model, we construct a novel scheme that adapts both the number of workers and the computation load throughout the run-time of the algorithm. Consequently, we improve the convergence speed of distributed SGD while significantly reducing the computation load, at the expense of a slight increase in communication load.

\end{abstract}
\section{Introduction}
\input{chapters/Intro}

\section{Preliminaries and System Model} \label{sec:prelim}
\input{chapters/preliminaries}
\input{chapters/system_model}

\section{Strategy and Main Results} \label{sec:strategy}

To optimize the speed of convergence, we group iterations into multiple stages $\stage \in \mathbb{N}$, where each stage is characterized by a certain number $\batchcnt_\stage$ of workers to wait for and a certain scale parameter $\batchscale_\stage$ for the batch size. We declare a new stage upon changing $\batchcnt$ and/or $\batchscale$. The time when switching from stage $\stage$ to $\stage+1$ is denoted by $\tc$, where $t_0 \define 0$. Consequently, the duration of a stage is $\tc-\tp$. Knowing the previous switching time $\tp$ and the parameters for the next stage, i.e., $\kn$ and $\betan$, we follow the steps of \cite{KasHanna2020} to determine the optimal switching times in theory, and by this to analyze possible gains of the proposed scheme. In contrast to \cite{KasHanna2020}, the order statistics and the error floors now also depend on $\beta$. The result is given in \cref{thm:sw_times}. For ease of notation, we omit the dependencies on $\batchcnt_\stage$ and $\batchscale_\stage$ and refer to the order statistics $\mu_{\kc:\workercnt}(\betac)$ as $\mukcur$. We define $\phi_\stage \define \batchcnt_\stage \batchscale_\stage$ to reflect the batch size in stage $\stage$.

\begin{theorem} \label{thm:sw_times}
Given $\tp$, the switching time $\tc$ is given by
\begin{align*}
    \tc &= \tp + \frac{\mukcur}{\alpha} \big(- \log(\mukcur \lr \lipschitz \gradvar (\phi_{\stage+1} - \phi_\stage)) \nonumber \\ & +\log(\muknext \! - \! \mukcur) + \log(\phi_{\stage+1}(2\convexity \phi_\stage  \batchsize e(\tp) \! - \! \lr \lipschitz \gradvar)) \big),
\end{align*} where $\alpha = -\log(1-\lr\convexity)$.
\begin{proof}
The proof follows similar steps as in \cite{KasHanna2020} and is given in \cref{app:proof_sw_times} for completeness.
\end{proof}

\end{theorem}

The question remaining is how to choose $\kn$ and $\betan$ depending on $\kc$ and $\betac$, which we answer in the following. %
The expectation of the order statistics is minimal when $\batchscale = \frac{1}{\batchsize}$ (i.e., only one sample is used) and $\batchcnt = 1$ (i.e., the main node waits for only one worker). The speed of convergence mainly depends on two factors: i) the order statistics of the workers' response times that depends on $\batchcnt$ and $\batchscale$, and ii) the error floor influenced by the product $\batchcnt \batchscale$.

As the exponential distribution is highly concentrated on the left, increasing the batch size is better than increasing the number of workers to wait for. This can be seen in \cref{prop:renyi}, where choosing $\kn=\kc+1$ to result in $\phi_{\stage+1}>\phi_\stage$ is strictly worse than choosing $\betan=(\kc+1)\betac/\kc$. Thus, we start the algorithm with $\batchcnt_1=1$.

Once the algorithm reaches the stationary phase, the quantity $\phi_\stage$ has to be increased to increase the effective mini-batch size. For a certain $\kc$, we gradually increase $\betac$ in steps of $\frac{1}{\batchsize}$, i.e., $\betan = \betac + \frac{1}{\batchsize}$, until $\betan=1$. Then, to further decrease the error, $\kc$ has to be incremented.\footnote{The increments of $\batchcnt$ and $\batchscale$ can be freely chosen to trade off the number of stages (and thus stationary phase detections) with the possible gain in convergence speed.} The parameter $\betan$ can be decreased accordingly such that $\phi_{\stage+1} > \phi_{\stage}$. Under this constraint, the optimal parameter $\betan$ has to be determined to maximize the convergence speed based on the system parameters. One approach is to incorporate the order statistics in \eqref{eq:convergence} to optimize convergence with respect to time \cite{KasHanna2020}. However, the derivations are intractable with more complex order statistics. Hence, we propose in \cref{thm:optimal_beta} to maximize the error as function of iterations while minimizing the expectation of the order statistics. %

\begin{theorem} \label{thm:optimal_beta}
When switching from waiting for $\kc$ workers to $\kn > \kc$, we define $\beta_{\min} \define \frac{1}{\batchsize} \! \cdot \! \left\lceil\frac{\kc \batchsize}{\kn}\right\rceil$. To obtain the optimal $\betan$, we solve the following optimization problem
\begin{align}
    \beta_{\text{opt}} &= \argmax_{\beta_{\min} \leq \batchscale \leq 1} \frac{\vert\partial \expdevk[\kn] / \partial \iter\vert}{\mu_{\kn:\workercnt}(\batchscale)} \define \argmax_{\beta_{\min} \leq \batchscale \leq 1} O, \! \nonumber
\end{align}
where to calculate $\partial \expdevk[\kn] / \partial \iter$ according to \eqref{eq:convergence}, we replace $\iter$ by $(\timevar_{\stage} - \tp)/\mukcur$ and $F(\initvec)$ by $F(\model_{\tp})$ to account for the state of convergence. The parameter $\betan$ is then obtained by rounding\footnote{If the result for $\kn>\kc$ is $\betan=1$, this could indicate that a larger value should be chosen for $\kn$.}, i.e., $\betan = \lceil \batchsize \batchscale_{\text{opt}} \rceil / \batchsize$.
If the objective $O$ is concave, the solution can be obtained by solving the unconstrained optimization and subsequently \emph{clipping} $\lceil \batchsize \batchscale_{\text{opt}} \rceil / \batchsize$ to $\betan \in [\batchscale_{\min},1]$. We obtain $\batchscale_{\text{opt}}$ by solving \eqref{eq:betaoptgeneral} for $\batchscale$ and choosing the value that maximizes $O$:
\begin{align}
    \!\!\!\! \frac{\partial \mu_{\kn:\workercnt}(\batchscale)}{\partial \batchscale} \batchscale (\batchscale \kn  - \phi_\stage) + \phi_\stage (\mukcur - \mu_{\kn:\workercnt}(\batchscale))  = 0 \! \label{eq:betaoptgeneral}
\end{align}
\end{theorem}
\begin{proof}
    The proof is given in \cref{sec:proof_optimal_beta}.
\end{proof}

If multiple solutions exist, one should choose the smallest $\betac$ as it yields the lowest computation effort.

\newcommand{\helper}{\ensuremath{\Psi}}
\begin{corollary} \label{cor:beta_opt_simple}
Under the model in \cref{def:simple_model}, for $\kn>\kc$, the unconstrained optimization in \cref{thm:optimal_beta} is concave in $\batchscale$. Solving \eqref{eq:betaoptgeneral} yields two possible values for $\batchscale$, i.e.,
\begin{equation*}
    \batchscale_{1,2} = \frac{\phi_\stage}{\kn} \left(1 \pm \sqrt{1-\frac{\kn}{\kc} \cdot \mukcurprimeone / \muknextprimeone}\right).
\end{equation*}
The value $\batchscale_{\text{opt}} \in \{\batchscale_1, \batchscale_2\}$ is chosen such that the objective $O$ in \cref{thm:optimal_beta} is maximized. Through clipping and rounding, we obtain $\betan = \max\{\lceil(\kc \batchsize)/\kn\rceil/\batchsize, \min\{1, \lceil \batchsize \batchscale_{\text{opt}}\rceil/\batchsize\}\}$.
\begin{proof}[Sketch of Proof]
The proof follows from inserting the expectation of the order statistics in \eqref{eq:betaoptgeneral}, simplifying, and finding the root of the resulting polynomial that maximizes $O$. A detailed proof is given in \cref{app:proof_beta_opt_simple}.
\end{proof}
\end{corollary}

For the general model in \cref{def:general_model}, a numerical solution can be obtained for $\betan$ when $\kn>\kc$. This requires the order statistics of the model, which we give in \cref{thm:order_stats}.
\begin{theorem} \label{thm:order_stats}
For $\comprate \neq \commrate$, under the model in \cref{def:general_model} the order statistics of $\rvcomb$ is given by %
\begin{align*}
    \mu_{\waitcnt:\clustersize}^{(2)} = &\sum_{\waitidx=\waitcnt}^\clustersize \sum_{\tmpa=0}^\waitidx \sum_{\tmpb=0}^{\tmpa+\clustersize-\waitidx} \sum_{\tmpc=0}^\tmpb \binom{\clustersize}{\waitidx} \binom{\waitidx}{\tmpa}  \binom{\tmpa+\clustersize-\waitidx}{\tmpb} \binom{\tmpb}{\tmpc} \cdot \\ &(-1)^{\tmpa+\tmpc+1} \frac{\commrate^\tmpa}{(\commrate - \comprate)^\tmpa \alpha} + \constcomm + \constcomp,
\end{align*} where $\alpha \define \commrate(\tmpa+\clustersize-\waitidx-\tmpb+\tmpc) + \comprate(\tmpb-\tmpc)$.
For $\comprate = \commrate$, we have the order statistics of an Erlang distribution $E(2,\comprate)$.
\begin{proof}
The proof is provided in \cref{app:proof_order_stats}.
\end{proof}
\end{theorem}

\section{Numerical Results} \label{sec:simulations}
\input{chapters/simulations_isit}

\section{Proof of \cref{thm:optimal_beta}} \label{sec:proof_optimal_beta}
\input{chapters/proof_optimal_beta_isit}

\section{Conclusion}
\input{chapters/conclusion}

\bibliographystyle{IEEEtran}
\bibliography{literature,paperpile,Old_Refs}

\newpage

\appendix

\subsection{Proof of \cref{thm:sw_times}}
\label{app:proof_sw_times}
\input{chapters/proof_switching_times}

\subsection{Proof of \cref{cor:beta_opt_simple}} \label{app:proof_beta_opt_simple}
\input{chapters/proof_beta_simple}

\subsection{Proof of \cref{thm:order_stats}}
\label{app:proof_order_stats}
\input{chapters/proof_order_stats}

\subsection{Additional Experiments} \label{app:additional_simulations}

We provide an additional set of experiments with the same setup used in \cref{sec:simulations}. To model the communication and computation delays, we use the delay model introduced in \cref{def:general_model}. In those experiments, we compare our scheme to the ones in \cite{Dutta2018} and \cite{KasHanna2020}. For the scheme of \cite{Dutta2018}, we chose three parameter sets $(\batchcnt,\beta) \in \{(1,0.2),(5,1),(10,1)\}$, such that the performance for all other values of $(\batchcnt,\beta)$ fall in between the extremes.

To instantiate the savings brought by our scheme, we choose three different regimes: 1) computation delays dominate over the communication delays; 2) computation and communication delays are comparable; and 3) communication delays dominate over the computation delays. We plot the convergence as a function of time for all schemes and all regimes.

\begin{figure}[H]
    \centering
    \includegraphics[width=\linewidth]{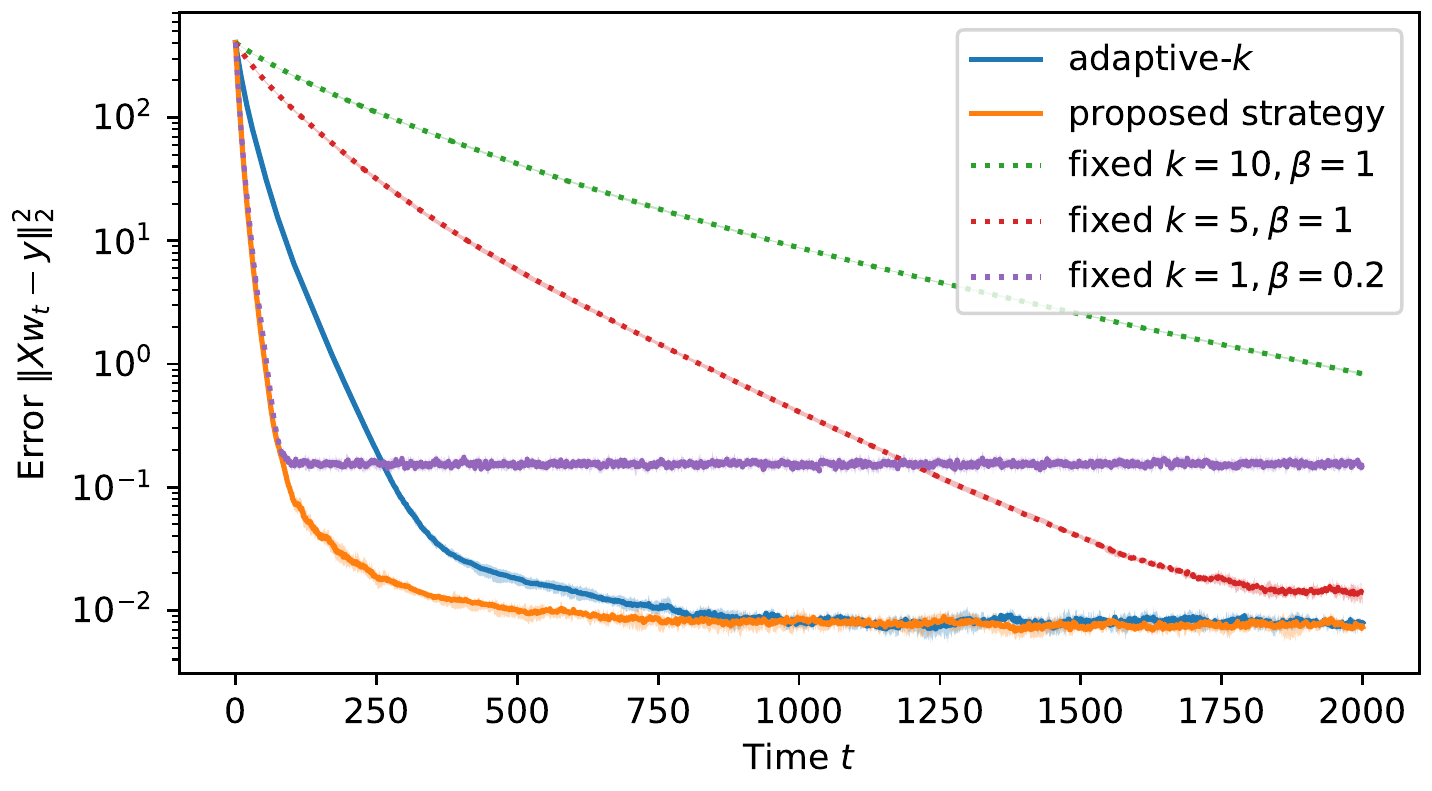}
    \caption{Error over time. $\lambda_y = 1$, $\lambda_x = 100$}
    \label{fig:error_sc}
\end{figure}

As shown by \cref{fig:error_sc}, the introduced strategy provides the best runtime improvements over the scheme in \cite{KasHanna2020} in cases where communication outweighs computation.

\begin{figure}[H]
    \centering
    \includegraphics[width=\linewidth]{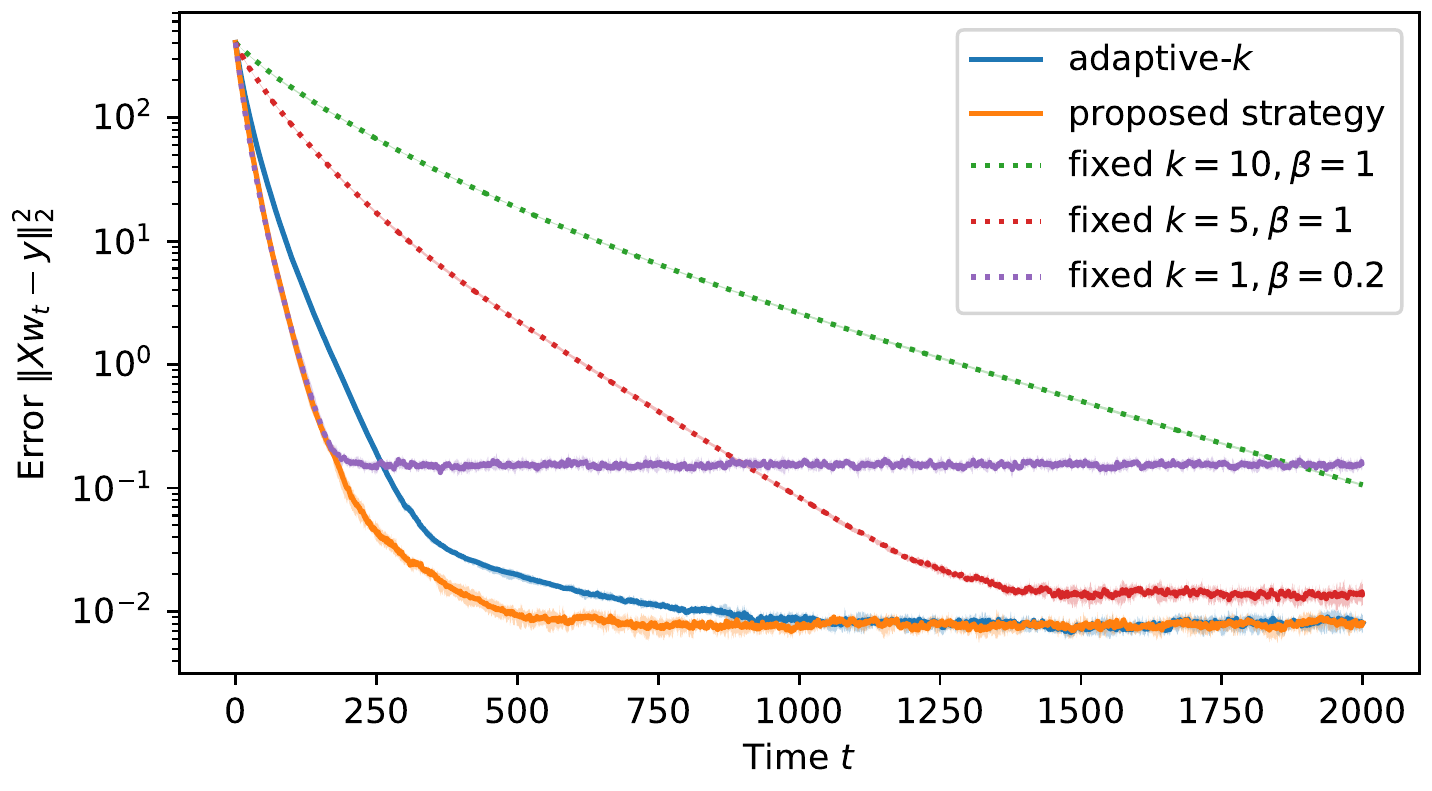}
    \caption{Error over time. $\lambda_y = 20$, $\lambda_x = 5/3$}
    \label{fig:error}
\end{figure}

\cref{fig:error} provides results for a parameter set in which communication efforts are time-wise comparable to computation spendings. Despite computation being still by a factor of $12$ faster than communication, we can observe notable improvements in terms of runtime.

\begin{figure}[H]
    \centering
    \includegraphics[width=\linewidth]{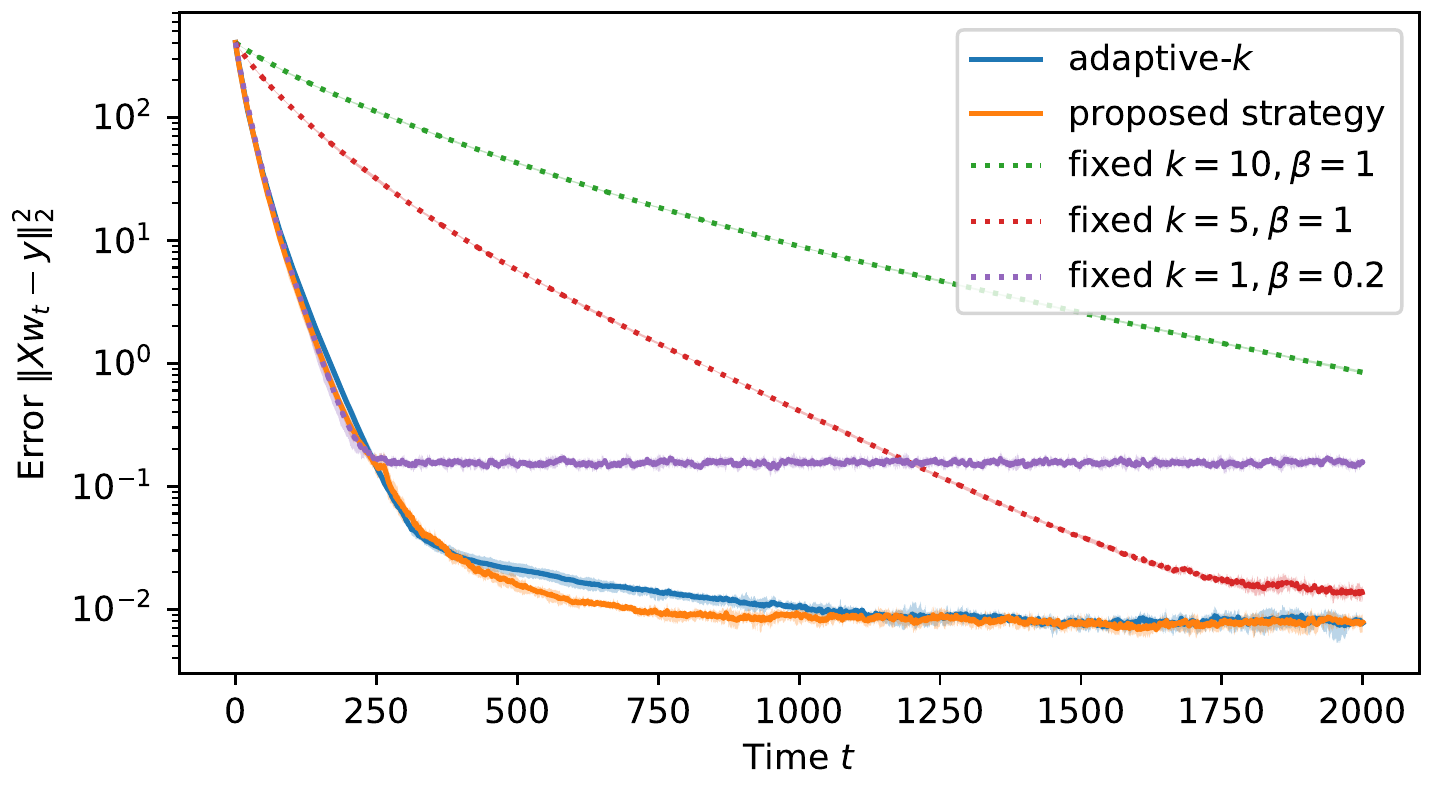}
    \caption{Error over time. $\lambda_y = 100$, $\lambda_x = 1$}
    \label{fig:error_lc}
\end{figure}

Only when computation is significantly faster, i.e., by factor $100$ in \cref{fig:error_lc}, the introduced algorithm does not improve the runtime. Hence, as predicted by the theory, we observe that our scheme enjoys the highest speed up in the first regime and no speed up in the third one.

\begin{figure}[H]
    \centering
    \includegraphics[width=\linewidth]{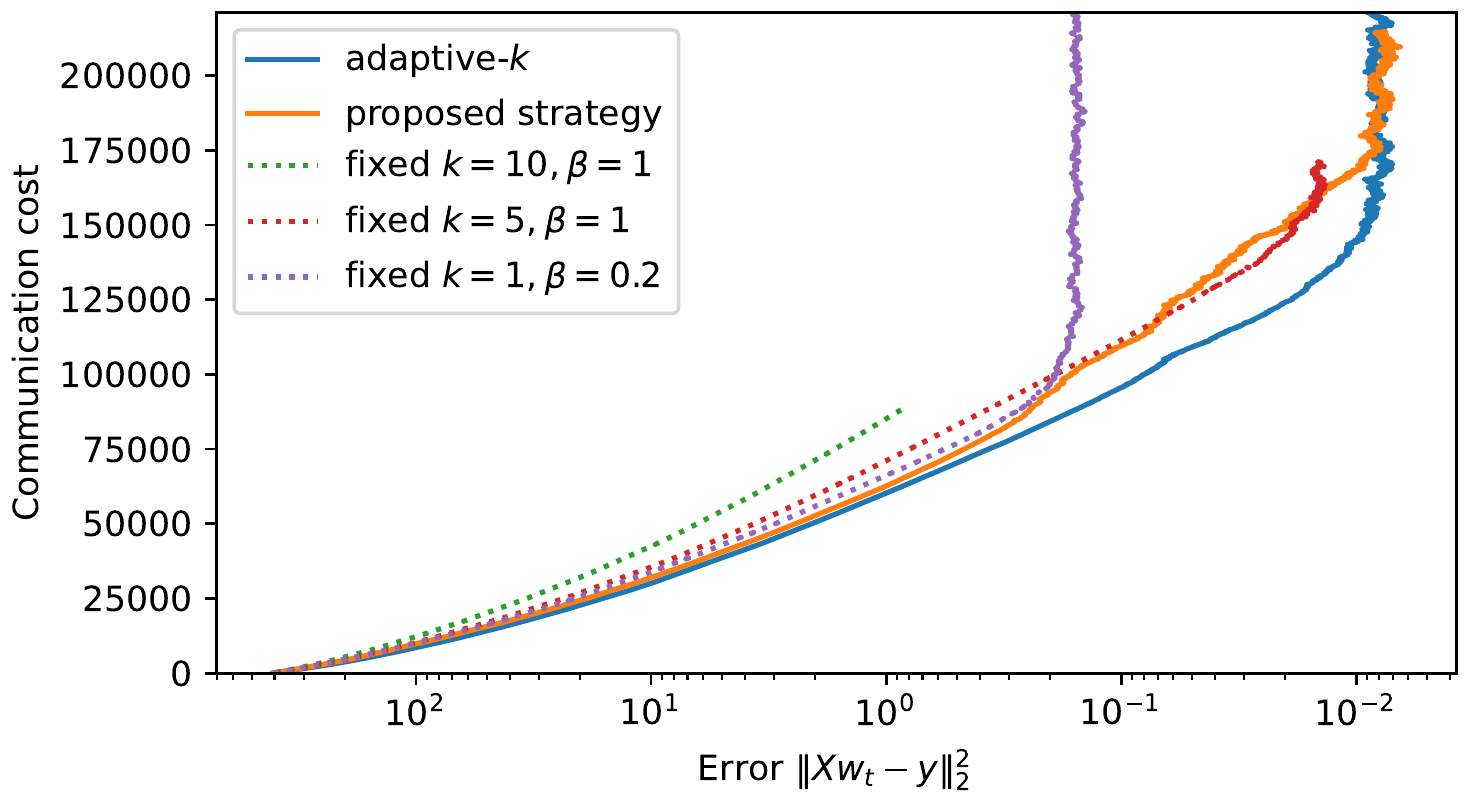}
    \caption{Communication cost over error.}
    \label{fig:comm_sc}
\end{figure}

We further analyze the communication (cf. \cref{fig:comm_sc}) and computation cost (cf. \cref{fig:comp_sc}) for all schemes. For clarity of presentation, we only depict the results for the third regime in which communication costs are dominating (i.e., no runtime improvement compared to \cite{KasHanna2020}) and clip the computation and communication costs to $40000$ and $210000$, respectively. For the other two regimes, the same qualitative behavior is observed. For all regimes, our scheme has a higher communication cost and a lower computation cost compared to the schemes on \cite{Dutta2018} and \cite{KasHanna2020}.

\begin{figure}[H]
    \centering
    \includegraphics[width=\linewidth]{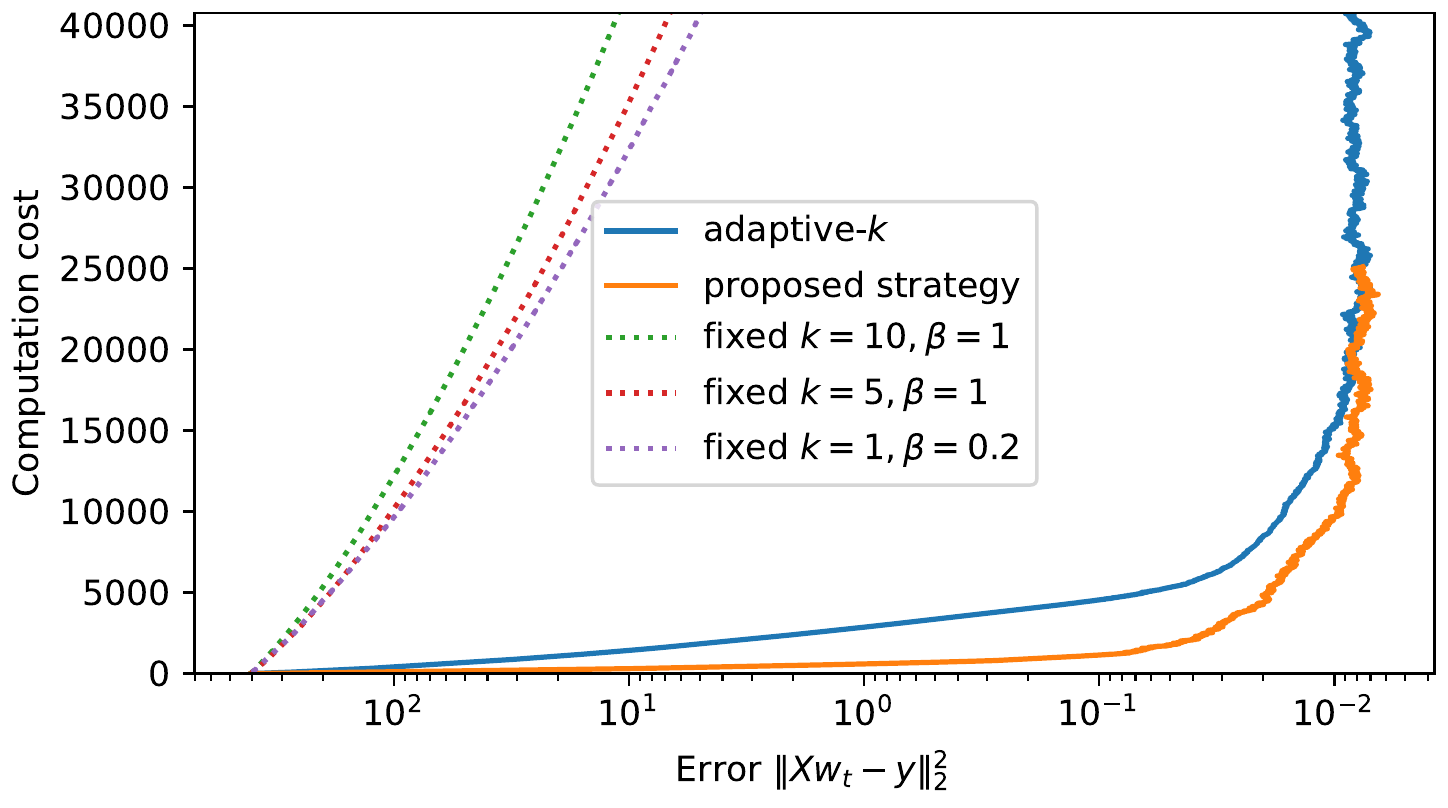}
    \caption{Computation cost over error.}
    \label{fig:comp_sc}
\end{figure}

\end{document}

%% file: acronyms.tex
\newacronym{sgd}{SGD}{stochastic gradient descent}
\newacronym{gd}{GD}{gradient descent}
\newacronym{cdf}{CDF}{cumulative distribution function}

%% file: commands.tex
\newcommand{\br}[1]{\ensuremath{\left(#1\right)}}

\DeclareMathOperator*{\argmax}{arg\,max}
\DeclareMathOperator*{\argmin}{arg\,min}

\newcommand{\define}{\ensuremath{\coloneqq}}

\newcommand{\samplesize}{\ensuremath{d}}
\newcommand{\samplecnt}{\ensuremath{v}}
\newcommand{\sampleidx}{\ensuremath{l}}

\newcommand{\samplevec}{\ensuremath{\mathbf{x}_\sampleidx}}
\newcommand{\samplematrix}{\ensuremath{\mathbf{X}}}
\newcommand{\workermatrix}[1][\armidx]{\ensuremath{\samplematrix_{#1}}}
\newcommand{\samplelabel}{\ensuremath{y_\sampleidx}}
\newcommand{\samplelabelvec}{\ensuremath{\mathbf{y}}}

\newcommand{\model}{\ensuremath{\mathbf{w}}}
\newcommand{\modelopt}{\ensuremath{\model_\star}}

\newcommand{\initvec}{\ensuremath{\model_0}}
\newcommand{\lossfct}[1][]{\ensuremath{F(\samplematrix, \samplelabelvec, \model_{#1})}}
\newcommand{\lossfctsample}[1][]{\ensuremath{F(\samplevec, \samplelabel, \model_{#1})}}

\newcommand{\lossfctiter}[1][\iter]{\ensuremath{F(\batchcnt, \batchscale, #1)}}
\newcommand{\lr}{\ensuremath{\eta}}
\newcommand{\lipschitz}{\ensuremath{L}}
\newcommand{\gradvar}{\ensuremath{\sigma^2_{\text{grad}}}}
\newcommand{\convexity}{\ensuremath{c}}
\newcommand{\batchsize}{\ensuremath{s}}
\newcommand{\batchcnt}{\ensuremath{k}}
\newcommand{\batchscale}{\ensuremath{\beta}}
\newcommand{\stage}{\ensuremath{\tau}}

\newcommand{\workeridx}{\ensuremath{i}}

\newcommand{\comm}{\ensuremath{X}}
\newcommand{\comp}{\ensuremath{Y}}
\newcommand{\comb}{\ensuremath{Z}}
\newcommand{\rvcomm}{\ensuremath{\comm_{\workeridx}}}
\newcommand{\rvcomp}{\ensuremath{\comp_{\workeridx}}}
\newcommand{\rvcomb}{\ensuremath{\comb_{\workeridx}}}
\newcommand{\constcomm}{\ensuremath{x}}
\newcommand{\constcomp}{\ensuremath{y\batchscale}}
\newcommand{\constcompnoscale}{\ensuremath{y}}

\newcommand{\commrate}{\ensuremath{\lambda_{x}}}
\newcommand{\comprate}{\ensuremath{\lambda_{y}/\batchscale}}
\newcommand{\compratenoscale}{\ensuremath{\lambda_{y}}}

\newcommand{\waitcnt}{\ensuremath{\batchcnt}}
\newcommand{\waitidx}{\ensuremath{j}}
\newcommand{\clustersize}{\ensuremath{\workercnt}}

\newcommand{\tmpa}{\ensuremath{\rho}}
\newcommand{\tmpb}{\ensuremath{\tau}}
\newcommand{\tmpc}{\ensuremath{\xi}}
\newcommand{\muk}{\ensuremath{\mu_{\waitcnt:\clustersize}(\batchscale)}}

\newcommand{\tc}{\ensuremath{t_\stage}}
\newcommand{\tp}{\ensuremath{t_{\stage-1}}}
\newcommand{\kc}{\ensuremath{\batchcnt_\stage}}
\newcommand{\kn}{\ensuremath{\batchcnt_{\stage+1}}}

\newcommand{\betac}{\ensuremath{\batchscale_\stage}}
\newcommand{\betan}{\ensuremath{\batchscale_{\stage+1}}}
\newcommand{\mukcur}{\ensuremath{\mu_{\stage}}}
\newcommand{\muknext}{\ensuremath{\mu_{\stage+1}}}
\newcommand{\mukcurone}{\ensuremath{\mu_{\stage}^{(1)}}}
\newcommand{\muknextone}{\ensuremath{\mu_{\stage+1}^{(1)}}}
\newcommand{\mukcurprime}{\ensuremath{\frac{\partial \mukcur}{\partial \betac}}}
\newcommand{\muknextprime}{\ensuremath{\frac{\partial \muknext}{\partial \betan}}}
\newcommand{\mukcurprimeone}{\ensuremath{\frac{\partial \mu_{\stage}^{(1)}}{\partial \betac}}}
\newcommand{\muknextprimeone}{\ensuremath{\frac{\partial \mu_{\stage+1}^{(1)}}{\partial \betan}}}

\newcommand{\workercnt}{\ensuremath{n}}

\newcommand{\timevar}{\ensuremath{t}}
\newcommand{\iter}{\ensuremath{j}}

\newcommand{\expdev}[1][\iter]{\ensuremath{E\br{\batchcnt, \batchscale, #1}}}
\newcommand{\expdevk}[1][\batchcnt]{\ensuremath{E\br{#1, \batchscale, \iter}}}

\newcommand{\itersamples}{\ensuremath{\mathcal{S}_\iter}}

%% file: chapters/Intro.tex
Tremendous amount of data is generated daily. Processing and analyzing this data requires distributing the computations to several nodes. Different forms of distributed computing are being explored. The two main settings are the main node/worker setting and the federated learning setting. In federated learning, clients generate their own data and work collaboratively with the help of a federator to train a machine learning model on their collective data. We refer the interested reader to \cite{kairouz2021advances} and references within for more details.

In this work, we focus on the main node/worker setting. The main node possesses a tremendous amount of data and wants to run a gradient descent-based machine learning algorithm on this data. To speed up the computation process, the main node wants to offload the computation to $n$ workers who are assigned smaller computation tasks that are run in parallel.

The main challenge of distributed computing is tolerating stragglers, i.e., slow or unresponsive workers, whose presence may delay the process \cite{DB13}. In addition to straggler tolerance, privacy of the main node's data and security against malicious workers trying to corrupt the computations is considered in the literature, e.g., \cite{ulukus2022private,d2020gasp,hofmeister2022secure,li2020coded} and references within. In this work, we focus solely on straggler-tolerance.%

To understand the effect of stragglers, consider the naive way of distributing the computation. The main node partitions its data and sends one part to each worker who runs a local computation and sends their result back to the main node. The main node needs to aggregate \emph{all} the workers' results to obtain the desired computation. This naive strategy assigns the smallest computation possible to each worker. However, it requires the main node to wait for the slowest worker which may outweigh the benefit of distributing the computations \cite{DB13}. The effect of stragglers can be mitigated by assigning redundant computation tasks to the workers, e.g., \cite{TLDK17,YA18,raviv2017gradient,lee2018speeding,ferdinand2018anytime,yu2018lagrange,kiani2018exploitation,chen2018draco,karakus2017straggler,halbawi2017improving,DCG16,KS18,fahim2017optimal,BPR17,Amiri2019,Ozfatura2020,kadhe2020communication} and waiting for a subset of them to send their results to compute the exact gradient. Nevertheless, adding redundancy increases the workload, hence the computation time, of each worker. Computing an estimate of the gradient in the presence of stragglers is considered in \cite{charles2017approximate,MRM18,bitar2020stochastic,wang2019erasurehead,wang2019fundamental,glasgow2021approximate,Ozfatura2021,sarmasarkar2021gradient}. The quality of the estimate of the gradient, and thus the convergence rate, depends on the number of stragglers and the chosen redundancy.

In distributed gradient descent (GD), mitigating the stragglers is as simple as ignoring\footnote{This holds true if the workers' data is generated from the same distribution.} them \cite{CPMBJ16}. Under some assumptions that we explain in the sequel, ignoring the stragglers is equivalent to the well-known mini-batch stochastic gradient descent (SGD) algorithm \cite{RM51}. SGD is a non-distributed relaxation of GD in which the main node is required to run the gradient descent algorithm on a subset (called batch) of the data. SGD and GD are iterative algorithms\footnote{Readers who are not familiar with GD algorithms are referred to Section~\ref{sec:prelim}.}. In a non-distributed setting, SGD is introduced to reduce the amount of computations required per iteration at the expense of requiring more iterations to \emph{converge}, i.e., to reach a desired result. The number of iterations required for SGD to converge is inversely proportional to the size of the batch on which the gradient is computed at every iteration. Murata \cite{Murata} showed that for a fixed learning rate, single-node SGD algorithms go through a transient phase and a stationary phase. In the former, the output of the algorithm approaches the optimal result exponentially fast in the number of iterations. Whereas, in the latter the output of the algorithm oscillates around the optimal result. Using a decreasing step size over the iterations \cite{pesme2020convergence,Ge2019} avoids oscillating around the optimal result. This, however, leads to a long transient phase and hence a lower convergence rate. The authors of \cite{Smith2018} show that it is beneficial to increase the mini-batch size instead of decreasing the learning rate.

In distributed SGD, a similar trade-off holds. 
Ignoring more workers results in a faster iteration, but gives a smaller effective batch size (hence requires more iterations). The works of \cite{CPMBJ16,Dutta2018,KasHanna2020,hanna2022adaptive} are the closest to our work. In \cite{CPMBJ16}, the authors study the effect of ignoring a small fraction of the workers and analyze the total time required by distributed SGD to converge. They show that there is an optimal number of workers to ignore at each iteration, equivalently an optimal number of tolerated stragglers, that is a fixed parameter throughout the algorithm. We shall call this strategy the \emph{fastest-$k$} strategy. The name stems from waiting for $k$ workers out of $n\geq k$ available ones, which reduces the overall run-time. In \cite{KasHanna2020}, the authors show that it is beneficial to adaptively choose the value of $\batchcnt$ that optimizes the convergence rate of the machine learning algorithm to further reduce the overall run-time. The intuition is that the main node kick-starts the algorithm on small batches. When a larger batch size is needed to improve the accuracy of the model, the main node decreases the straggler tolerance. We refer to this strategy as \emph{adaptive-$k$} strategy.

For adaptive SGD algorithms, heuristics are crucial to detect the stationary phase, i.e., when the error floor is reached. The authors of \cite{KasHanna2020} proposed to use Pflug's method \cite{Pflug1990}. This approach performs reasonably well for specific tasks and parameter ranges, but is sensitive to the learning rate \cite{Chee2018,pesme2020convergence}. The authors of \cite{Chee2018} proposed to use implicit updates for more robustness. Alternative criteria are considered in \cite{pesme2020convergence}, which we slightly modify and use in our simulations.

In \cite{Dutta2018,KasHanna2020} the response time of a worker is modeled by one random variable that captures the time spent by the worker on the computation and on communicating the result to the main node. In contrast, we model the response time as the sum of two independent random variables for communication time and computation time as in \cite{schlegel2022privacy,Amiri2019,dhakal2019coded}. %
We show that in certain regimes, it is beneficial to increase/decrease the computation load of the workers before changing the straggler tolerance. We provide an algorithm with a smaller run-time and lower overall computation effort than the \emph{adaptive-$k$} strategy \cite{KasHanna2020}. We compare our strategy to that of \cite{KasHanna2020}, since in this line of research \emph{adaptive-$k$} outperforms prior works. %

%% file: chapters/preliminaries.tex
\textbf{Notation.} Vectors and matrices are denoted by bold lower-case and upper-case letters, respectively. Let $\tmpa$ be an integer, then $[\tmpa] \define \{1, \dots, \tmpa\}$. Let $\tmpb, \tmpc \in \mathbb{R}$, then $[\tmpb, \tmpc]$ is the interval of all reals between and including $\tmpb$ and $\tmpc$.

%% file: chapters/system_model.tex
In this work, we focus on the main node/worker setting, in which a central entity runs an optimization procedure based on input data consisting of multiple samples with respective labels. A sample consists of multiple features. To reduce the computation time required to finish the optimization, computations are distributed to a total of $\workercnt$ worker machines, indexed by $\workeridx \in [\workercnt]$. Let $\samplematrix \in \mathbb{R}^{\samplecnt \times \samplesize}$ be the matrix that contains the $\samplecnt$ row-wise stacked features for every sample $\samplevec \in \mathbb{R}^{\samplesize}$, and let $\samplelabelvec \in \mathbb{R}^{\samplecnt}$ be the vector containing all labels $\samplelabel \in \mathbb{R}$. The main node distributes disjoint subsets of $\samplematrix$ of size $\batchsize = \frac{\samplecnt}{\workercnt}$ to every worker $\workeridx \in [\workercnt]$, denoted by $\workermatrix \in \mathbb{R}^{\batchsize \times \samplesize}$.

We iteratively seek to find a model $\modelopt$ that minimizes an additively separable global objective function $\lossfct \define \sum_{\sampleidx = 1}^{\samplecnt} F(\samplevec, \samplelabel, \model)$. %
To optimize the convergence behavior, we particularly focus on distributed SGD. To mitigate the effect of stragglers, at every iteration the main node only waits for the computations of $\batchcnt$ workers. The batch size $\batchsize$ per worker is adaptively scaled by $\batchscale \in \{\frac{1}{\samplecnt}, \frac{2}{\samplecnt}, \dots, \frac{\samplecnt-1}{\samplecnt}, 1\}$ and the number of workers $\batchcnt$ to be waited for at every iteration. That is, at iteration $\iter$ every worker $\workeridx$ uses a random batch of size $\batchscale \batchsize$ from $\workermatrix$. Let $\itersamples$ denote the sample indices of the $\batchcnt$-fastest workers at iteration $\iter$, the objective (loss) function in iteration $\iter$ then reads as $\lossfctiter \define \frac{1}{\batchcnt \batchscale \batchsize} \sum_{\sampleidx \in \itersamples} F(\samplevec, \samplelabel, \model_\iter)$. With learning rate $\lr$, the model update reads as
$$
    \model_{\iter+1} = \model_{\iter} - \lr \nabla \lossfctiter. %
$$
Consider an optimization algorithm based on SGD. An upper bound for the expected deviation of the actual loss $\lossfctiter$ from the optimal loss $F^\star \define F(\samplematrix, \samplelabelvec, \modelopt)$, further referred to as the error $\expdev \define \mathbb{E}[\lossfctiter - F^\star]$, was derived in \cite{Bottou2018,Dutta2018}. 
Let $\lipschitz$ be the Lipschitz constant, $\gradvar$ an upper bound on the variance of the gradient of the loss function, and $\convexity$ the convexity parameter \cite{Bottou2018}, then we have
\begin{equation}
    \!\!\!\!\!\expdev \! \leq \! \underbrace{\frac{\eta \lipschitz \gradvar}{2 \convexity \batchcnt \batchsize \batchscale}}_{\text{error floor}} \! + \! \underbrace{(1-\eta \convexity)^{\iter} \Big(\! F(\initvec) \! - \! F^\star \! - \! \frac{\eta \lipschitz \gradvar}{2 \convexity \batchcnt \batchsize \batchscale}}_{\text{transient phase}} \! \Big).\!\! \label{eq:convergence}
\end{equation}
Ultimately, we want to optimize the decay of $\expdev$ with respect to time by varying the parameters $\batchscale$ and $\batchcnt$. The time dependency results from the duration of an iteration $\iter$.

Because of the scaling parameter $\batchscale$, the workers' batches can possibly consist of only one sample. If the delay from communication significantly outweighs the computation time due to small task sizes, iterations get costly in terms of convergence over time. For a fair comparison to existing literature, we thus differentiate between communication and computation time. We measure the per iteration cost of communication as $\workercnt+\batchcnt$, i.e., the sum of model vectors communicated to the workers and back to the main node, and the cost of computation as $\batchscale\batchsize$. %
In \cref{def:simple_model}, we model the computation time of the workers as independently and identically distributed (i.i.d.) exponentially distributed random variables with common rate \cite{Liang2014,Lee2016}. The communication time is modeled as an equal constant for all workers.
This model mainly serves to provide tractable and accessible theoretical results in the sequel.
\begin{definition}[Simplified Model] \label{def:simple_model}
The response time for worker $\workeridx$ is modeled by random variable $\rvcomb = \rvcomm + \rvcomp$, where the communication time $\rvcomm = \constcomm = \text{const.}$ and the computation time $\rvcomp$ is modeled by a scaled and shifted exponential random variable, i.e., $\rvcomp \sim \exp(\comprate) + \constcompnoscale$, where $\constcompnoscale = \text{const.}$
\end{definition}

We denote by $\mu_{\batchcnt:\workercnt}$ the expectation of the $\batchcnt$-th order statistics of $\workercnt$ i.i.d. random variables $\rvcomb$, where $\workeridx \in [\workercnt]$.

\begin{proposition} \label{prop:renyi}
The model in \cref{def:simple_model} results in the well-known order statistics for $\batchcnt$ out of $\workercnt$ workers responding \cite{Renyi1953-vo}:
\begin{equation*}
    \mu_{\batchcnt:\workercnt}^{(1)}(\beta) = \frac{\batchscale}{\compratenoscale} \sum_{j=\workercnt-\batchcnt+1}^{\workercnt}  \frac{1}{j} + \constcompnoscale + \constcomm.
\end{equation*}
\end{proposition}

More generally, we model in \cref{def:general_model} computation and communication as independent shifted exponential random variables. %
To give directions how to run the proposed algorithm under this model, we provide the order statistics in \cref{thm:order_stats}.
\begin{definition}[Generalized Model] \label{def:general_model}
The response time for worker $\workeridx$ is modeled by $\rvcomb = \rvcomm + \rvcomp$, where the communication time $\rvcomm \sim \constcomm + \exp(\commrate)$, with $\constcomm = \text{const.}$, and the computation time $\rvcomp \sim \constcomp + \exp(\comprate)$, with $\constcompnoscale = \text{const.}$
\end{definition}

%% file: chapters/simulations_isit.tex
To compare the proposed strategy to the \emph{adaptive-$k$} strategy in \cite{KasHanna2020}, %
we first evaluate the theoretical impact for $\workercnt=50$ workers. We consider the Lipschitz constant $\lipschitz = 2$, an upper bound on the variance of the gradient of the loss function $\gradvar = 10$, and convexity constant $\convexity = 1$. We use the model in \cref{def:simple_model} and determine the switching times and the optimal parameters based on \cref{thm:sw_times,cor:beta_opt_simple}, respectively. 
The results are evaluated upon reaching an error of $\num{1e-3}$ according to \eqref{eq:convergence}.
We compare the resulting gain in runtime and computation cost and the communication overhead to the \emph{adaptive-$k$} strategy proposed in \cite{KasHanna2020} for different parameter sets, i.e., $\compratenoscale, \constcomm \in [0.05, 20]$. The parameter $\constcompnoscale$ is set to zero.

\textbf{Runtime Improvement.} The runtime of the proposed strategy is strictly smaller than that of the \emph{adaptive-$k$} strategy. The gain depends on the setting and is shown in \cref{fig:runtime}. Clearly, the improvement in runtime is largest when communication time is negligible compared to computation time. In this regime, scaling the task sizes has the most impact as compared to \emph{adaptive-$k$}. If communication is dominating over computation, the benefits of varying the batch size are negligible. For a small batch size, more iterations are required to converge and the gain in expected time per iteration is negligible.
\vspace{-.35cm}
\begin{figure}[H]
    \centering
    \includegraphics[width=.8\linewidth]{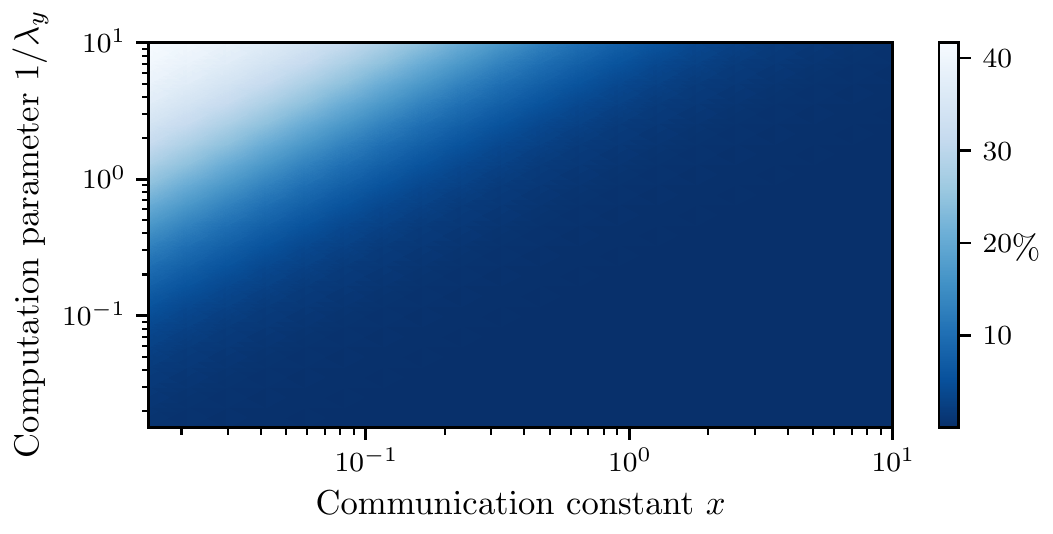}
    \caption{Runtime improvements for $\rvcomb \sim \exp(\lambda_y/\beta) + \constcomm$. Gains, indicated by light colors, are observed for small communication constants and large computation parameters.}
    \label{fig:runtime}
\end{figure}
\vspace{-.2cm}
\textbf{Communication Overhead.} %
Since the size of the result per iteration does not depend on the batch size used per worker, our strategy requires at least the same amount of communication as \emph{adaptive-$k$}. The results are depicted in \cref{fig:communication}. In regimes where the proposed strategy is most beneficial over \emph{adaptive-$k$}, this overhead is in the order of $17\%$.
\vspace{-.35cm}
\begin{figure}[H]
    \centering
    \includegraphics[width=.8\linewidth]{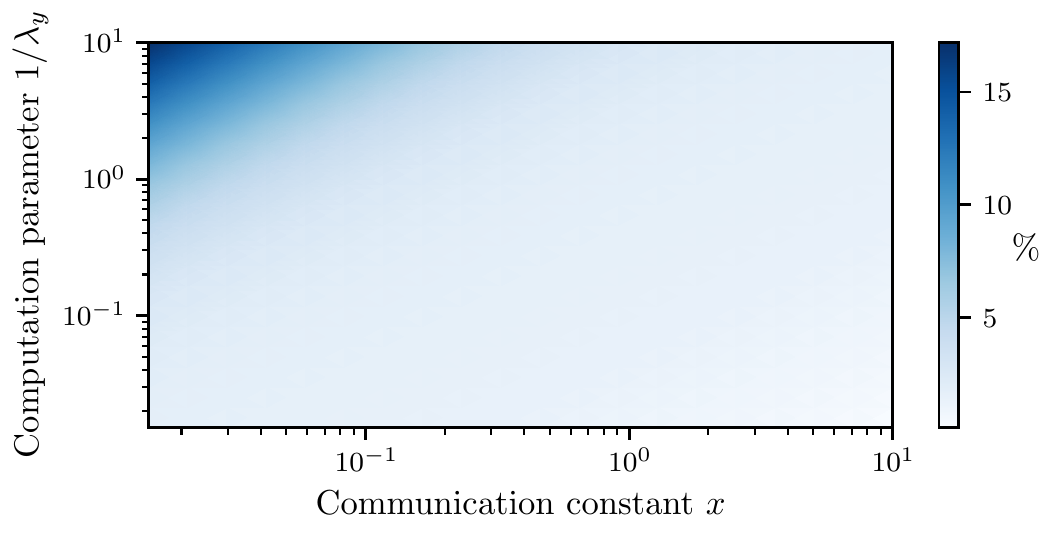}
    \caption{Communication overhead for $\rvcomb \sim \exp(\lambda_y/\beta) +  \constcomm$. Large overheads are indicated by dark colors.}
    \label{fig:communication}
\end{figure}
\vspace{-.2cm}

\textbf{Reduction in Computation.} The advantages of our strategy are twofold. While improving the runtime of the algorithm, it decreases the computation cost. That is because the decrease of task sizes has a larger impact than the increase in number of iterations. \cref{fig:computation} shows the experimental results.
\begin{figure}
    \centering
    \includegraphics[width=.8\linewidth]{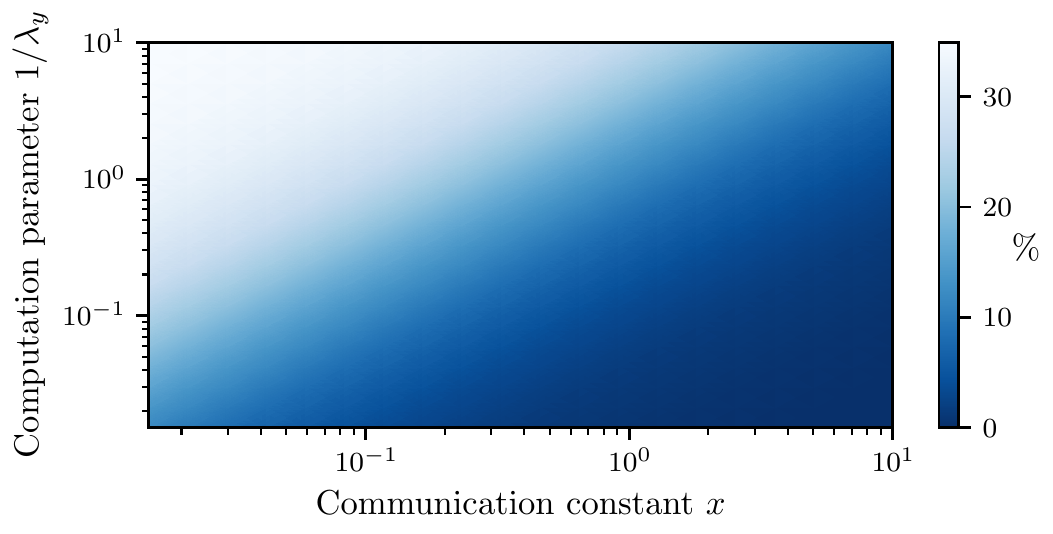}
    \caption{Reduction in computation for $\rvcomb \sim \exp(\lambda_y/\beta) + \constcomm$. Gains, indicated by light colors, are observed for small communication constants and large computation parameters.}
    \label{fig:computation}
    \vspace{-.6cm}
\end{figure}
Extensive numerical experiments show that all the effects mentioned above scale as the number of samples increases. Intuitively, the larger the number of samples $\batchsize$ per worker, the larger the difference between our strategy and \emph{adaptive-$k$}. %

\textbf{Simulations.} We simulate the runtime benefits compared to \emph{adaptive-$k$} by closely following the linear regression setting in \cite{KasHanna2020}. The features of $\samplematrix$ are generated uniformly at random from $[100]$ and the labels $\samplelabelvec$ are chosen uniformly at random from $[10]$. We employ $\workercnt=20$ workers to compute the partial gradients of the loss function $\lossfctsample = (\samplevec \model - \samplelabel)^2$ for a total of $\samplecnt = 400$ samples. The response time of each worker is modeled by a shifted exponential random variable according to \cref{def:simple_model} with $\compratenoscale=1$ and $\constcomm=0.01$. We use $k \leq 10$ and limit the set of possible $\batchscale$'s to $\batchscale \in \{0.2, 0.4, 0.6, 0.8, 1\}$. We use an adapted version of the convergence diagnostics in \cite{pesme2020convergence} to determine the switching times, which is more robust as compared to Pflug's method \cite{Pflug1990}. All presented results are averaged over $100$ simulation runs and depicted in \cref{fig:simulations}.

\vspace{-.2cm}
\begin{figure}[H]
    \centering
    \includegraphics[width=0.7\linewidth]{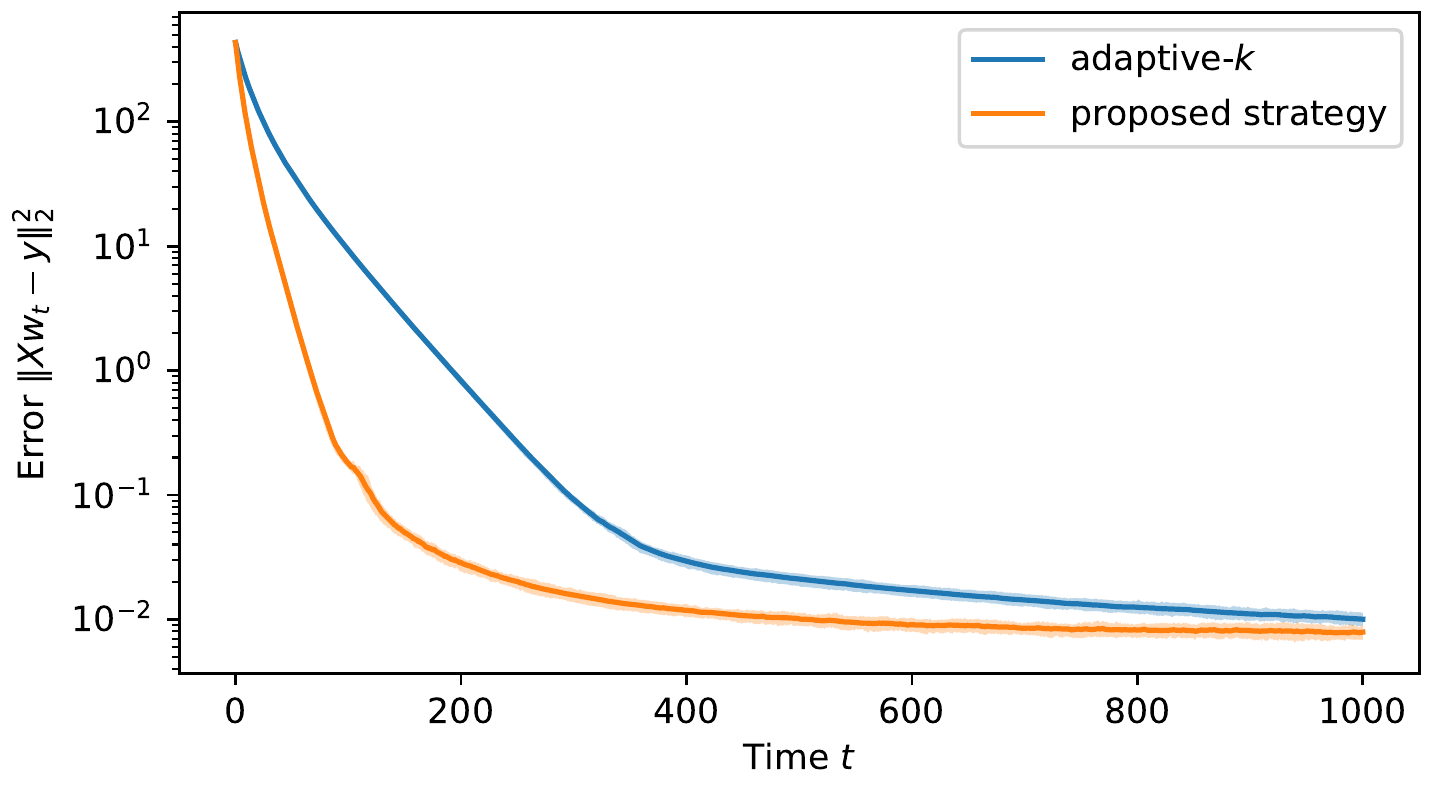}
    \caption{Numerical simulations with $80\%$ quantiles.}
    \label{fig:simulations}
    \vspace{-.4cm}
\end{figure}
One can observe that the proposed algorithm roughly halves the runtime. Further, for a desired error of $\num{2e-2}$, the proposed strategy reduces the computation effort by %
$\approx 59.9\%$ at the expense of increasing the communication load by $\approx 15.7\%$. %
Additional simulations for the delay model introduced in \cref{def:general_model} are provided in \cref{app:additional_simulations}.

%% file: chapters/proof_optimal_beta_isit.tex
For calculating the switching time $\tc$, we require the knowledge of the defining parameters $\kc$ and $\betac$ of the current and the next stage, i.e., $\kn > \kc$ and $1 \geq \betan \geq \beta_{\min}$. However, $\betan$ is yet unknown. To get the best possible $\beta_\text{opt}$, we maximize the objective $O$ at the iteration $j$ of switching with respect to $\betan$.
We replace the current iteration index $\iter$ in \eqref{eq:convergence} by the number of iterations ran in this stage starting from a reference iteration $\iter^\prime$ and change the reference error from $F(\initvec)$ to $F(\model_{\iter^\prime})$. Hence, we obtain
\begin{align}
    \vert\partial E(\kn, \beta, j) / \partial \iter\vert = \alpha e^{-\alpha (\iter-\iter^\prime)} \bigg(e(\iter^\prime) - \frac{\eta \lipschitz \gradvar}{2 \convexity \phi_{\stage+1} \batchsize} \bigg),
    \label{eq:convergence_rate_proof}
\end{align}
with $\alpha \define - \log(1-\lr \convexity) > 0$ and $e(\iter^\prime) \define F(\model_{\iter^\prime}) - F^\star$ being the difference of the loss based on the model $\model_{\iter^\prime}$ at iteration $\iter^\prime$ compared to the optimal loss. Since we are interested in the convergence rate right after switching, we set $\iter=\iter^\prime$, i.e., $e(\iter^\prime) = e(\iter)$.
To obtain $e(\iter)$, as in \cite{KasHanna2020} we use \eqref{eq:convergence}, replace $\iter$ by the expected number of iterations $(\tc-\tp)/\mukcur$ ran in this stage and use the reference error $e(\tp) \define F(\model_{\tp})-F^\star$. Thus, the convergence rate is a function of the switching time $\tc$, which is again a function of $\betan$. From \cref{thm:sw_times} we know $\tc-\tp$, which is
\begin{align}
    &\!\!\!\tc - \tp = \nonumber \\
    &\!\!\!\!= \frac{\mukcur}{\alpha} \log \left( \frac{(\muknext \! - \! \mukcur) \phi_{\stage+1}(2\convexity \phi_\stage  \batchsize e(\tp) \! - \! \lr \lipschitz \gradvar)}{\mukcur \lr \lipschitz \gradvar (\phi_{\stage+1} - \phi_\stage)} \right)\!.\! \label{eq:swtimediff}
\end{align}
Thus, we obtain for the expected deviation from the optimal loss at switching iteration $\iter$
\begin{align}
    &e(\iter^\prime) = e(\iter) = E(\kc, \betac, \iter) = \nonumber \\
    &= \frac{\lr \lipschitz \gradvar}{2 \convexity \batchsize \phi_\stage} \! + \! e^{\frac{-(\tc - \tp)\alpha}{\mukcur}} \bigg(e(\tp) - \frac{\lr \lipschitz \gradvar}{2 \convexity \batchsize \phi_\stage} \bigg) \label{eqline:swtimediff_before} \\
    &= \frac{\lr \lipschitz \gradvar}{2 \convexity \batchsize \phi_\stage} \! + \! \frac{\mukcur \lr \lipschitz \gradvar (\phi_{\stage+1} - \phi_\stage)  \Big(e(\tp) - \frac{\lr \lipschitz \gradvar}{2 \convexity \batchsize \phi_\stage} \Big)}{(\muknext \! - \! \mukcur) \phi_{\stage+1}(2\convexity \phi_\stage  \batchsize e(\tp) \! - \! \lr \lipschitz \gradvar)} \!\!\label{eqline:swtimediff_after} \\
    &= \frac{\lr \lipschitz \gradvar}{2 \convexity \batchsize \phi_\stage} + \frac{\lr \lipschitz \gradvar}{2\convexity\batchsize \phi_\stage} \frac{\mukcur (\phi_{\stage+1} - \phi_\stage) \Big(e(\tp) - \frac{\lr \lipschitz \gradvar}{2 \convexity \batchsize \phi_\stage} \Big)}{(\muknext \! - \! \mukcur) \phi_{\stage+1}\Big(e(\tp) \! - \! \frac{\lr \lipschitz \gradvar}{2\convexity\batchsize \phi_\stage}\Big)} \nonumber \\
    &= \frac{\lr \lipschitz \gradvar}{2 \convexity \batchsize \phi_\stage} \bigg( 1 + \frac{(\phi_{\stage+1} - \phi_\stage) \mukcur}{(\muknext \! - \! \mukcur) \phi_{\stage+1}} \bigg), \nonumber
\end{align}
where from \eqref{eqline:swtimediff_before} to \eqref{eqline:swtimediff_after} we inserted \eqref{eq:swtimediff}. Using this result and that $\alpha e^{-\alpha (\iter-\iter^\prime)}$ is independent of $\betan$, the objective function $O$ can be simplified to
\begin{align*}
    O &= \frac{e(\iter^\prime) - \frac{\eta \lipschitz \gradvar}{2 \convexity \phi_{\stage+1} \batchsize}}{\muknext} = \frac{\frac{\lr \lipschitz \gradvar}{2 \convexity \batchsize \phi_\stage} \bigg( 1 + \frac{(\phi_{\stage+1} - \phi_\stage) \mukcur}{(\muknext \! - \! \mukcur) \phi_{\stage+1}} \bigg) - \frac{\eta \lipschitz \gradvar}{2 \convexity \phi_{\stage+1} \batchsize}}{\muknext} \\
    &= \frac{1}{\muknext} \bigg( \frac{1}{\phi_\stage} \bigg( 1 + \frac{(\phi_{\stage+1} - \phi_\stage) \mukcur}{(\muknext \! - \! \mukcur) \phi_{\stage+1}} \bigg) - \frac{1}{\phi_{\stage+1}} \bigg) \\
    &= \frac{1}{\muknext} \bigg( \frac{1}{\phi_\stage} - \frac{1}{\phi_{\stage+1}} + \frac{(\phi_{\stage+1} - \phi_\stage) \mukcur}{(\muknext \! - \! \mukcur) \phi_\stage \phi_{\stage+1}} \bigg) \\
    &= \frac{1}{\muknext} \bigg( \frac{\phi_{\stage+1}-\phi_\stage}{\phi_\stage \phi_{\stage+1}} + \frac{(\phi_{\stage+1} - \phi_\stage) \mukcur}{(\muknext \! - \! \mukcur) \phi_\stage \phi_{\stage+1}} \bigg) \\
    &= \frac{1}{\muknext} \frac{\phi_{\stage+1}-\phi_\stage}{\phi_\stage \phi_{\stage+1}} \bigg( 1 + \frac{\mukcur}{(\muknext \! - \! \mukcur)} \bigg) \\
    &= \frac{1}{\muknext} \frac{\phi_{\stage+1}-\phi_\stage}{\phi_\stage \phi_{\stage+1}} \bigg( \frac{\muknext \! - \! \mukcur}{\muknext \! - \! \mukcur} + \frac{\mukcur}{(\muknext \! - \! \mukcur)} \bigg) \\
    &= \frac{1}{\muknext} \frac{\phi_{\stage+1}-\phi_\stage}{\phi_\stage \phi_{\stage+1}} \bigg( \frac{\muknext}{\muknext \! - \! \mukcur} \bigg) \\
    &= \frac{\phi_{\stage+1}-\phi_\stage}{\phi_\stage \phi_{\stage+1}} \bigg( \frac{1}{\muknext \! - \! \mukcur} \bigg).
\end{align*}
Since $\phi_\stage$ does not depend on $\betan$, the optimization problem consequently results as
\begin{align*}
    \beta_{\text{opt}} &= \argmax_{\beta_{\min} \leq \betan \leq 1} O = \argmin_{\beta_{\min} \leq \betan \leq 1} \frac{\phi_{\stage+1} \muknext \! - \! \phi_{\stage+1} \mukcur}{\phi_{\stage+1}-\phi_\stage}. \\
    &= \argmin_{\beta_{\min} \leq \betan \leq 1} O^\prime,
\end{align*}
where the scaled and inverted objective $O$ is referred to as $$O^\prime \define \frac{1}{O \cdot \phi_\stage} = \frac{\phi_{\stage+1} \muknext \! - \! \phi_{\stage+1} \mukcur}{\phi_{\stage+1}-\phi_\stage}.$$ The parameter $\betan$ has to be a multiple of $\frac{1}{\batchsize}$, and hence we require rounding to $\betan = \lceil \batchsize \batchscale_{\text{opt}} \rceil / \batchsize$.
If the objective is a concave function, we can obtain the solution by setting the derivative to zero and clipping $\lceil \batchsize \batchscale_{\text{opt}} \rceil / \batchsize$ to $\betan \in [\batchscale_{\min},1]$. Hence, if the objective is a concave function in $\betan$, we only require that
\begin{align*}
    \frac{\partial O^\prime}{\partial \betan} = &\frac{\partial}{\partial \betan} \bigg( \frac{\phi_{\stage+1} \muknext \! - \! \phi_{\stage+1} \mukcur}{\phi_{\stage+1}-\phi_\stage} \bigg) = 0.
\end{align*}
We define $f(\betan) \define \phi_{\stage+1} \muknext \! - \! \phi_{\stage+1} \mukcur$ and $g(\betan) \define \phi_{\stage+1}-\phi_\stage$ and apply the quotient rule so that
\begin{align*}
    \frac{\partial O^\prime}{\partial \betan} &= \frac{\partial (f(\betan)/g(\betan))}{\partial \betan} = \\
    &= \frac{\frac{\partial f(\betan)}{\partial \betan} \cdot g(\betan) - \frac{\partial g(\betan)}{\partial \betan} \cdot f(\betan)}{g(\betan)^2}.
\end{align*} 
\renewcommand{\mukcurprime}{\ensuremath{\mukcur^\prime}}
\renewcommand{\muknextprime}{\ensuremath{\muknext^\prime}}
For clarity of presentation, we further define $\mukcurprime \define \frac{\partial \mukcur}{\partial \betac}$ and $\muknextprime = \frac{\partial \muknext}{\partial \betan}$. Ignoring the denominator $g(\betan)^2>0$, we obtain
\begin{align}
    0 &= \frac{\partial f(\betan)}{\partial \betan} \cdot g(\betan) - \frac{\partial g(\betan)}{\partial \betan} \cdot f(\betan) \nonumber \\
    &= \begin{aligned}[t] &\left(\kn \muknext + \phi_{\stage+1} \muknextprime - \kn \mukcur \right) (\phi_{\stage+1}-\phi_\stage)   \nonumber \\
    &- (\phi_{\stage+1} \muknext - \phi_{\stage+1} \mukcur) \kn \end{aligned}  \nonumber \\
    &= \begin{aligned}[t] &\kn \muknext \phi_{\stage+1} + \phi_{\stage+1}^2 \muknextprime - \kn \mukcur \phi_{\stage+1}  \nonumber \\
    &- \kn \muknext \phi_\stage - \phi_\stage \phi_{\stage+1} \muknextprime + \kn \mukcur \phi_\stage  \nonumber \\
    &- \phi_{\stage+1} \muknext \kn - \phi_{\stage+1} \mukcur \kn \end{aligned} \nonumber \\
    &= \phi_{\stage+1}^2 \muknextprime - \phi_\stage \kn (\muknext - \mukcur) - \phi_\stage \phi_{\stage+1} \muknextprime. \label{eqline:derivative_numerator}
\end{align}
This is equivalent to solving
\begin{align*}
    \muknextprime \betan (\betan \kn - \phi_\stage) - \phi_\stage (\muknext - \mukcur) = 0,
\end{align*}
which concludes the proof.

%% file: chapters/conclusion.tex
We proposed an adaptive distributed SGD strategy that varies the size of the computational task assigned to the workers and the straggler tolerance per iteration to improve the runtime of the algorithm. While the expected cost of communication is slightly increased, the overall computation effort is reduced significantly. The main benefits of this strategy are obtained when communication is fast and cheap compared to computation, e.g., in data centers with shared clusters. %

%% file: chapters/proof_switching_times.tex
For the proof of \cref{thm:sw_times}, we follow similar lines as in \cite{KasHanna2020}. Let $\alpha \define -\log(1-\eta c)$, then by expressing the iteration $\iter$ as the expected number of iterations after time $\timevar$ as $\frac{\timevar}{\muk}$ we can express \eqref{eq:convergence} considering the parameters of stage $\stage$, i.e., $\betac$ and $\kc$ as
\begin{align*}
    E(\kc, \betac, \timevar) &= \frac{\eta \lipschitz \gradvar}{2 \convexity \kc \batchsize \betac} \! + \! (1-\eta \convexity)^{\frac{\timevar}{\mukcur}} \Big(\! F(\initvec) \! - \! F^\star \! - \! \frac{\eta \lipschitz \gradvar}{2 \convexity \kc \batchsize \betac}\! \Big) \\
    & = \frac{\eta \lipschitz \gradvar}{2 \convexity \kc \batchsize \betac} \! + \! e^{\frac{-\alpha \timevar}{\mukcur}} \Big(\! F(\initvec) \! - \! F^\star \! - \! \frac{\eta \lipschitz \gradvar}{2 \convexity \kc \batchsize \betac}\! \Big).
\end{align*}

Let $e(\timevar) \define F(\model_\timevar) - F^\star$ be the difference of the loss based on the model $\model_\timevar$ at time $\timevar$ compared to the optimal loss. Let further $\timevar^\prime \leq \timevar$ be a reference time, then the convergence rate, i.e., the decrease of the error in \eqref{eq:convergence} with respect to time, can be calculated as follows:
\begin{align}
    \left\vert \frac{\partial E(\kc, \betac, \timevar)}{\partial \timevar} \right\vert &= \frac{\alpha}{\mukcur} e^{\frac{-\alpha (\timevar-\timevar^\prime)}{\mukcur}} \bigg(\! e(\timevar^\prime) \! - \! \frac{\eta \lipschitz \gradvar}{2 \convexity \kc \batchsize \betac}\! \bigg) \label{eq:convergence_rate}
\end{align}

For the optimal switching time at when to advance from stage $\stage$ to stage $\stage+1$, the following criterion has to be fulfilled:
\begin{equation}
    \label{eq:cond}
    \left\vert \frac{\partial E(\kn, \betan, \timevar)}{\partial \timevar} \right\vert > \left\vert \frac{\partial E(\kc, \betac, \timevar)}{\partial \timevar} \right\vert
\end{equation}

Thereby, we use $\tp$ as reference time in stage $\stage$, i.e., for $\tc \geq \timevar \geq \tp$ and $\tc$ as reference time in stage $\stage+1$, i.e., for $\timevar_{\stage+1} \geq \timevar \geq \tc$.
The value for $e(\tc)$ is given by
\begin{align}
    e(\tc) = \frac{\eta \lipschitz \gradvar}{2 \convexity \kc \batchsize \betac} \! + \! e^{\frac{-\alpha (\tc-\tp)}{\mukcur}} \Big(\! e(\tp) \! - \! \frac{\eta \lipschitz \gradvar}{2 \convexity \kc \batchsize \betac}\! \Big). \label{eq:etc}
\end{align}
We express the rates in \cref{eq:cond} by \eqref{eq:convergence_rate} with the respective reference times, evaluate the resulting equation at $\timevar=\tc$ and substitute $e(\tc)$ by \eqref{eq:etc}. Simplifying yields
\begin{align*}
    e^{\frac{-\alpha (\tc-\tp)}{\mukcur}} &(\muknext-\mukcur) \bigg(e(\tp)-\frac{\lr \lipschitz \gradvar}{2 \convexity \batchsize \phi_\stage} \bigg) \\
    & \leq \mukcur \bigg( \frac{\lr \lipschitz \gradvar}{2 \convexity \batchsize \phi_\stage} - \frac{\lr \lipschitz \gradvar}{2 \convexity \batchsize \phi_{\stage+1}} \bigg).
\end{align*}
Further simplifying and solving for $\tc$ concludes the proof.

%% file: chapters/proof_beta_simple.tex
We will first proof that the objective $O$ in \cref{thm:optimal_beta} is concave in $\beta$ for the model in \cref{def:simple_model}. In the proof of \cref{thm:optimal_beta}, the numerator resulting from the quotient rule for the derivative of $O^\prime$ was determined. Considering this particular model, we have $f(\betan) \define \phi_{\stage+1} \muknextone \! - \! \phi_{\stage+1} \mukcurone$ and $g(\betan) \define \phi_{\stage+1}-\phi_\stage$. We use the simplification in \eqref{eqline:derivative_numerator} as the numerator of the quotient rule and simplify the expression using the model in \cref{def:simple_model}.
\renewcommand{\mukcurprimeone}{\ensuremath{\mu_{\stage}^{{(1)}^\prime}}}
\renewcommand{\muknextprimeone}{\ensuremath{\mu_{\stage+1}^{(1)^\prime}}}
For clarity of presentation, we define $\mukcurprimeone \define \frac{\partial \mu_{\stage}^{(1)}}{\partial \betac}$ and $\muknextprimeone \define \frac{\partial \mu_{\stage+1}^{(1)}}{\partial \betan}$. Considering the model in \cref{def:simple_model}, it holds that $\muknext = \betan \muknextprimeone + \constcomm$. Thus, using \eqref{eqline:derivative_numerator} we have
\begin{align*}
    &\frac{\partial f(\betan)}{\partial \betan} \cdot g(\betan) - \frac{\partial g(\betan)}{\partial \betan} \cdot f(\betan) = \nonumber \\
   &= \muknextprimeone \betan (\betan \kn^2 - \phi_\stage \kn) - \phi_\stage \kn (\muknextone - \mukcur) \\
   &= \begin{aligned}[t] \muknextprimeone \betan &(\betan \kn^2 - \phi_\stage \kn) \\ &- \phi_\stage \kn (\betan \muknextprimeone + \constcomm - \mukcur) \end{aligned} \\
   &= \begin{aligned}[t] \betan^2 \muknextprimeone &\kn^2 - \betan \left( \muknextprimeone \phi_\stage \kn + \phi_\stage \kn \muknextprimeone \right) \\
   &- \phi_\stage \kn (\constcomm - \mukcur) \end{aligned} \\
   &= \begin{aligned}[t] \betan^2 \muknextprimeone \kn^2 - 2 \betan \muknextprimeone \phi_\stage \kn - \phi_\stage \kn (\constcomm - \mukcur) \end{aligned} \\
   &= \begin{aligned}[t] \phi_\stage \kn \muknextprimeone \left( \betan^2 \kn / \phi_\stage - 2 \betan - (\constcomm - \mukcur) / \muknextprimeone \right). \end{aligned}
\end{align*}

Thus, the derivative of the transformed objective $O^\prime$ w.r.t. $\betan$ results as
\begin{align*}
    \frac{\partial O^\prime}{\partial \betan} &= \frac{\phi_\stage \kn \muknextprimeone \left( \frac{\betan^2 \kn}{\phi_\stage} - 2 \betan - \frac{\constcomm - \mukcur}{\muknextprimeone} \right)}{(\phi_{\stage+1} - \phi_\stage)^2}.
\end{align*}

Since $\phi_\stage \kn \muknextprimeone$ does not depend on $\betan$, we omit this term to show concaveness of $O^\prime$ w.r.t. $\betan$. We thus calculate a scaled function of the second derivative of $O^\prime$ as
\begin{align}
    & \frac{\partial}{\partial \betan} \frac{\phi_\stage \kn \muknextprimeone \left( \frac{\betan^2 \kn}{\phi_\stage} - 2 \betan - \frac{\constcomm - \mukcur}{\muknextprimeone} \right)}{(\phi_{\stage+1} - \phi_\stage)^2} \nonumber \\
    &= \frac{(2\betan \kn / \phi_\stage - 2) (\phi_{\stage+1} - \phi_\stage)^2}{(\phi_{\stage+1} - \phi_\stage)^4} \nonumber \\ 
    &\,\,\,- \frac{\left(\frac{\betan^2 \kn}{\phi_\stage} - 2 \betan - \frac{\constcomm - \mukcur}{\muknextprimeone} \right) \cdot 2 (\phi_{\stage+1} - \phi_\stage) \kn}{(\phi_{\stage+1} - \phi_\stage)^4} \!\!\! \label{eqline:before_partial_exchange} \\ 
    &= \frac{2(\phi_{\stage+1} / \phi_\stage - 1) (\phi_{\stage+1} - \phi_\stage)}{(\phi_{\stage+1} - \phi_\stage)^3} \nonumber \\ 
    &\,\,\,- \frac{2\left(\phi_{\stage+1}^2 / \phi_\stage - 2 \phi_{\stage+1} + \betac \kn \mukcurprimeone / \muknextprimeone \right)}{(\phi_{\stage+1} - \phi_\stage)^3}, \label{eqline:after_partial_exchange}
\end{align}
where from \eqref{eqline:before_partial_exchange} to \eqref{eqline:after_partial_exchange} we used that $\mukcurone - \constcomm = \betac \mukcurprimeone$. Since by design it holds that $\phi_{\stage+1} - \phi_\stage > 0$, it suffices to show that
\begin{align*}
    \left(\frac{\phi_{\stage+1}}{\phi_\stage} - 1\right) &(\phi_{\stage+1} - \phi_\stage) \\
    &- \left(\frac{\phi_{\stage+1}^2}{\phi_\stage} - 2 \phi_{\stage+1} + \betac \kn \frac{\mukcurprimeone}{\muknextprimeone} \right) > 0.
\end{align*}
We simplify the expression and obtain
\begin{align*}
    & \begin{aligned}[t] \left(\frac{\phi_{\stage+1}}{\phi_\stage} - 1\right) &(\phi_{\stage+1} - \phi_\stage) \\
    &- \left(\frac{\phi_{\stage+1}^2}{\phi_\stage} - 2 \phi_{\stage+1} + \betac \kn \frac{\mukcurprimeone}{\muknextprimeone} \right) \end{aligned} \\
    &= \begin{aligned}[t] \phi_{\stage+1} &\left(\frac{\phi_{\stage+1}}{\phi_\stage} - 1\right) (1 - \frac{\phi_\stage}{\phi_{\stage+1}}) \\
    &- \phi_{\stage+1} \left(\frac{\phi_{\stage+1}}{\phi_\stage} - 2 + \frac{\kn}{\kc} \frac{\mukcurprimeone}{\muknextprimeone} \right) \end{aligned} \\
    &= \begin{aligned}[t] \phi_{\stage+1} &\left(\frac{\phi_{\stage+1}}{\phi_\stage} - \frac{\phi_\stage}{\phi_{\stage+1}} - 2\right) - \phi_{\stage+1} \left(\frac{\phi_{\stage+1}}{\phi_\stage} - 2\right) \\
    &+ \phi_{\stage} \left( \frac{\kn}{\kc} \frac{\mukcurprimeone}{\muknextprimeone} \right) \end{aligned} \\
    &= \phi_{\stage+1} \left(\frac{\phi_\stage}{\phi_{\stage+1}}\right) - \phi_{\stage} \left(\frac{\kn}{\kc} \frac{\mukcurprimeone}{\muknextprimeone} \right)  \\
    &= \phi_{\stage} \left(1-\frac{\kn}{\kc} \frac{\mukcurprimeone}{\muknextprimeone} \right) > 0,
\end{align*}
which is true since $\phi_\stage>0$ and
\begin{align*}
    1-\frac{\kn}{\kc} \frac{\mukcurprimeone}{\muknextprimeone} &> 0 \\
    \frac{\kn}{\kc} &< \frac{\muknextprimeone}{\mukcurprimeone} \\
    \frac{\kn}{\kc} &< \frac{\sum_{j=n-\kn+1}^n \frac{1}{j}}{\sum_{j=n-\kc+1}^n \frac{1}{j}} \\
    \frac{\kc + (\kn-\kc)}{\kc} &< \frac{\sum_{j=n-\kc+1}^n \frac{1}{j} + \sum_{j=n-\kn+1}^{n-\kc} \frac{1}{j}}{\sum_{j=n-\kc+1}^n \frac{1}{j}} \\
    1+ \frac{\kn-\kc}{\kc} &< 1+ \frac{\sum_{j=n-\kn+1}^{n-\kc}}{\sum_{j=n-\kc+1}^n \frac{1}{j}}, \\
    \frac{\kn-\kc}{\kc} &< \frac{\sum_{j=n-\kn+1}^{n-\kc} \frac{1}{j}}{\sum_{j=n-\kc+1}^{n} \frac{1}{j}}.
\end{align*}
This relation holds since $\sum_{j=n-\kn+1}^{n-\kc} \frac{1}{j}$ contains $\kn-\kc > 0$ summands, each of which are larger than each of the $\kc$ summands in $\sum_{j=n-\kc+1}^{n} \frac{1}{j}$. This concludes the proof for the concavity of the objective $O$ w.r.t $\betan$. To solve the optimization, we use the result of \cref{thm:optimal_beta} and the properties of the model given in \cref{def:simple_model} to solve
\begin{align*}
    & \betan^2 \kn \muknextprimeone - \betan \muknextprimeone \phi_\stage - \phi_\stage \left(\muknextone - \mukcurone\right) = \\
    &= \betan^2 \kn \muknextprimeone - \betan \muknextprimeone \phi_\stage - \phi_\stage \left(\betan \muknextprimeone + \constcomm - \mukcurone\right) \\
    &= \betan^2 \kn \muknextprimeone - 2 \betan \muknextprimeone \phi_\stage + \phi_\stage \betac \mukcurprimeone = 0,
\end{align*}
where we used that $\mukcurone - \constcomm = \betac \mukcurprimeone$. We divide by $\phi_\stage \muknextprimeone$ and obtain the following degree two polynomial in $\betan$ that we require to be zero:
\begin{align*}
    &\betan^2 \frac{\kn}{\phi_\stage} - 2 \betan \phi_\stage + \phi_\stage \betac \frac{\mukcurprimeone}{\muknextprimeone}
\end{align*}
Thus, the root of this polynomial that maximizes $O$ is the global optimizer, which concludes the proof.

%% file: chapters/proof_order_stats.tex
To proof \cref{thm:order_stats}, we first calculate the \ac{cdf} of $\rvcomb$.
\begin{align*}
    F_{\rvcomb} (z) &= \Pr(\rvcomm + \rvcomp \leq z) = \int_\mathbb{R} \int_{-\infty}^{z-x} f_{\rvcomm,\rvcomp} (x,y) dy dx \\
    &= \int_\mathbb{R} F_{\rvcomp} (z-x) f(x) dx \\
    &= \int_0^z (1-e^{-\comprate(z-x)}) \commrate e^{-\commrate x} dx \\
    &= \int_0^z \commrate e^{-\commrate x} - \commrate \int_0^z e^{x(\comprate - \commrate) - \comprate z} dx \\
    &= [-e^{-\commrate x}]^z_0 - \frac{\commrate}{\comprate - \commrate} [e^{x(\comprate - \commrate) - \comprate z}]_0^z \\
    &= -e^{-\commrate z} + 1 + \frac{\commrate}{\comprate - \commrate} (e^{-\comprate z} - e^{-\commrate z}).
\end{align*} 
Consider $\workercnt$ of these random variables, then the \ac{cdf} for their order statistics reads as
\begin{align*}
    &F_{\rvcomb(\waitcnt)}(z) = \sum_{\waitidx=\waitcnt}^\clustersize \binom{\clustersize}{\waitidx} (F_{\rvcomb}(z))^\waitidx (1-F_{\rvcomb}(z))^{\clustersize-\waitidx} \\
    &= \begin{aligned}[t]
    \sum_{\waitidx=\waitcnt}^\clustersize &\binom{\clustersize}{\waitidx} \left(1 - e^{-\commrate z} - \frac{\commrate}{\commrate - \comprate} (e^{-\comprate z} - e^{-\commrate z})\right)^\waitidx \cdot \\&\left(e^{-\commrate z} + \frac{\commrate}{\commrate - \comprate} (e^{-\comprate z} - e^{-\commrate z})\right)^{\clustersize-\waitidx} \end{aligned} \\
    &= \begin{aligned}[t] \sum_{\waitidx=\waitcnt}^\clustersize \binom{\clustersize}{\waitidx} \sum_{\tmpa=0}^\waitidx \binom{\waitidx}{\tmpa} &(-1)^\tmpa \Big(e^{-\commrate z} + \\ 
    &+ \frac{\commrate}{\commrate - \comprate} (e^{-\comprate z} - e^{-\commrate z})\Big)^{\tmpa + \clustersize-\waitidx} \end{aligned} \\
    &= \begin{aligned}[t] \sum_{\waitidx=\waitcnt}^\clustersize &\binom{\clustersize}{\waitidx} \sum_{\tmpa=0}^\waitidx \binom{\waitidx}{\tmpa} (-1)^\tmpa  \sum_{\tmpb=0}^{\tmpa+\clustersize-\waitidx} \binom{\tmpa+\clustersize-\waitidx}{\tmpb} (e^{-\commrate z})^{\tmpa + \clustersize - \waitidx - \tmpb}  \cdot \\
    &\left(\frac{\commrate}{\commrate - \comprate}\right)^\tmpb (e^{-\comprate z} - e^{-\commrate z})^{\tmpb} \end{aligned} \\
    &= \begin{aligned}[t] \sum_{\waitidx=\waitcnt}^\clustersize &\binom{\clustersize}{\waitidx} \sum_{\tmpa=0}^\waitidx \!\! \binom{\waitidx}{\tmpa} (-1)^\tmpa \sum_{\tmpb=0}^{\tmpa+\clustersize-\waitidx} \binom{\tmpa+\clustersize-\waitidx}{\tmpb} (e^{-\commrate z})^{\tmpa + \clustersize - \waitidx - \tmpb}  \cdot \\
    &\left(\frac{\commrate}{\commrate - \comprate}\right)^\tmpb \sum_{\tmpc=0}^\tmpb \binom{\tmpb}{\tmpc} (e^{-\comprate z})^{\tmpb-\tmpc} (-1)^\tmpc (e^{-\commrate z})^{\tmpc} \end{aligned} \\
    &= \begin{aligned}[t] \sum_{\waitidx=\waitcnt}^\clustersize &\sum_{\tmpa=0}^\waitidx \sum_{\tmpb=0}^{\tmpa+\clustersize-\waitidx} \sum_{\tmpc=0}^\tmpb \binom{\clustersize}{\waitidx} \binom{\waitidx}{\tmpa}  \binom{\tmpa+\clustersize-\waitidx}{\tmpb} \binom{\tmpb}{\tmpc} (-1)^{\tmpa+\tmpc} \cdot \\
    &\left(\frac{\commrate}{\commrate - \comprate}\right)^\tmpb  e^{-z(\commrate(\tmpa+\clustersize-\waitidx-\tmpb+\tmpc) + \comprate(\tmpb-\tmpc))}. \end{aligned} \\
\end{align*}
Next, we let $\alpha \define \commrate(\tmpa+\clustersize-\waitidx-\tmpb+\tmpc) + \comprate(\tmpb-\tmpc)$, and calculate the expectation of the order statistics of $\clustersize$ random variables as
\begin{align*}
    \mu_{\waitcnt:\clustersize} &= \mathbb{E}[F_{\rvcomb(\waitcnt)}(z)] = \int_0^{\infty}  z \frac{\partial {F_{\rvcomb(\waitcnt)(z)}}}{\partial z} dz \\
    &= \begin{aligned}[t] \sum_{\waitidx=\waitcnt}^\clustersize &\sum_{\tmpa=0}^\waitidx \sum_{\tmpb=0}^{\tmpa+\clustersize-\waitidx} \sum_{\tmpc=0}^\tmpb \binom{\clustersize}{\waitidx} \binom{\waitidx}{\tmpa}  \binom{\tmpa+\clustersize-\waitidx}{\tmpb} \binom{\tmpb}{\tmpc} \\ &(-1)^{\tmpa+\tmpc} \frac{\commrate^\tmpa (-\alpha)}{(\commrate - \comprate)^\tmpa} \int_0^\infty z e^{-z \alpha} \end{aligned} \\
    &= \begin{aligned}[t] \sum_{\waitidx=\waitcnt}^\clustersize &\sum_{\tmpa=0}^\waitidx \sum_{\tmpb=0}^{\tmpa+\clustersize-\waitidx} \sum_{\tmpc=0}^\tmpb \binom{\clustersize}{\waitidx} \binom{\waitidx}{\tmpa}  \binom{\tmpa+\clustersize-\waitidx}{\tmpb} \binom{\tmpb}{\tmpc} \\ &(-1)^{\tmpa+\tmpc+1} \frac{\commrate^\tmpa \alpha}{(\commrate - \comprate)^\tmpa} \frac{1}{\alpha^2} \end{aligned} \\
    &= \begin{aligned}[t] \sum_{\waitidx=\waitcnt}^\clustersize &\sum_{\tmpa=0}^\waitidx \sum_{\tmpb=0}^{\tmpa+\clustersize-\waitidx} \sum_{\tmpc=0}^\tmpb \binom{\clustersize}{\waitidx} \binom{\waitidx}{\tmpa}  \binom{\tmpa+\clustersize-\waitidx}{\tmpb} \binom{\tmpb}{\tmpc} \\ &(-1)^{\tmpa+\tmpc+1} \frac{\commrate^\tmpa}{(\commrate - \comprate)^\tmpa \alpha}. \end{aligned}
\end{align*}
To conclude the proof, we add to $\mu_{\waitcnt:\clustersize}$ the deterministic shift of the shifted exponential random variables $\rvcomp$ and $\rvcomm$.

%% file: arXiv.bbl
\begin{thebibliography}{10}
\providecommand{\url}[1]{#1}
\csname url@samestyle\endcsname
\providecommand{\newblock}{\relax}
\providecommand{\bibinfo}[2]{#2}
\providecommand{\BIBentrySTDinterwordspacing}{\spaceskip=0pt\relax}
\providecommand{\BIBentryALTinterwordstretchfactor}{4}
\providecommand{\BIBentryALTinterwordspacing}{\spaceskip=\fontdimen2\font plus
\BIBentryALTinterwordstretchfactor\fontdimen3\font minus
  \fontdimen4\font\relax}
\providecommand{\BIBforeignlanguage}[2]{{%
\expandafter\ifx\csname l@#1\endcsname\relax
\typeout{** WARNING: IEEEtran.bst: No hyphenation pattern has been}%
\typeout{** loaded for the language `#1'. Using the pattern for}%
\typeout{** the default language instead.}%
\else
\language=\csname l@#1\endcsname
\fi
#2}}
\providecommand{\BIBdecl}{\relax}
\BIBdecl

\bibitem{kairouz2021advances}
P.~Kairouz, H.~B. McMahan, B.~Avent, A.~Bellet, M.~Bennis, A.~N. Bhagoji,
  K.~Bonawitz, Z.~Charles, G.~Cormode, R.~Cummings \emph{et~al.}, ``Advances
  and open problems in federated learning,'' \emph{Foundations and Trends in
  Machine Learning}, vol.~14, no. 1--2, pp. 1--210, 2021.

\bibitem{DB13}
J.~Dean and L.~A. Barroso, ``The tail at scale,'' \emph{Communications of the
  ACM}, vol.~56, no.~2, pp. 74--80, 2013.

\bibitem{ulukus2022private}
S.~Ulukus, S.~Avestimehr, M.~Gastpar, S.~Jafar, R.~Tandon, and C.~Tian,
  ``Private retrieval, computing and learning: Recent progress and future
  challenges,'' \emph{IEEE Journal on Selected Areas in Communications}, 2022.

\bibitem{d2020gasp}
R.~G. D’Oliveira, S.~El~Rouayheb, and D.~Karpuk, ``{GASP} codes for secure
  distributed matrix multiplication,'' \emph{IEEE Transactions on Information
  Theory}, vol.~66, no.~7, pp. 4038--4050, 2020.

\bibitem{hofmeister2022secure}
C.~Hofmeister, R.~Bitar, M.~Xhemrishi, and A.~Wachter-Zeh, ``Secure private and
  adaptive matrix multiplication beyond the singleton bound,'' \emph{IEEE
  Journal on Selected Areas in Information Theory}, 2022.

\bibitem{li2020coded}
S.~Li, S.~Avestimehr \emph{et~al.}, ``Coded computing: Mitigating fundamental
  bottlenecks in large-scale distributed computing and machine learning,''
  \emph{Foundations and Trends in Communications and Information Theory},
  vol.~17, no.~1, pp. 1--148, 2020.

\bibitem{TLDK17}
R.~Tandon, Q.~Lei, A.~G. Dimakis, and N.~Karampatziakis, ``Gradient coding:
  Avoiding stragglers in distributed learning,'' in \emph{International
  Conference on Machine Learning}, 2017, pp. 3368--3376.

\bibitem{YA18}
M.~Ye and E.~Abbe, ``Communication-computation efficient gradient coding,'' in
  \emph{International Conference on Machine Learning}, 2018, pp. 5610--5619.

\bibitem{raviv2017gradient}
N.~Raviv, I.~Tamo, R.~Tandon, and A.~G. Dimakis, ``Gradient coding from cyclic
  {MDS} codes and expander graphs,'' \emph{IEEE Transactions on Information
  Theory}, vol.~66, no.~12, pp. 7475--7489, 2020.

\bibitem{lee2018speeding}
K.~Lee, M.~Lam, R.~Pedarsani, D.~Papailiopoulos, and K.~Ramchandran, ``Speeding
  up distributed machine learning using codes,'' \emph{IEEE Transactions on
  Information Theory}, vol.~64, no.~3, pp. 1514--1529, 2018.

\bibitem{ferdinand2018anytime}
N.~Ferdinand and S.~C. Draper, ``Anytime stochastic gradient descent: A time to
  hear from all the workers,'' in \emph{Allerton Conference on Communication,
  Control, and Computing}, 2018, pp. 552--559.

\bibitem{yu2018lagrange}
Q.~Yu, S.~Li, N.~Raviv, S.~M.~M. Kalan, M.~Soltanolkotabi, and S.~A.
  Avestimehr, ``Lagrange coded computing: Optimal design for resiliency,
  security, and privacy,'' in \emph{International Conference on Artificial
  Intelligence and Statistics}, 2019, pp. 1215--1225.

\bibitem{kiani2018exploitation}
S.~Kiani, N.~Ferdinand, and S.~C. Draper, ``Exploitation of stragglers in coded
  computation,'' in \emph{IEEE International Symposium on Information Theory},
  2018, pp. 1988--1992.

\bibitem{chen2018draco}
L.~Chen, H.~Wang, Z.~Charles, and D.~Papailiopoulos, ``{DRACO}:
  {B}yzantine-resilient distributed training via redundant gradients,'' in
  \emph{International Conference on Machine Learning}, J.~Dy and A.~Krause,
  Eds., vol.~80, 2018, pp. 903--912.

\bibitem{karakus2017straggler}
C.~Karakus, Y.~Sun, S.~Diggavi, and W.~Yin, ``Straggler mitigation in
  distributed optimization through data encoding,'' in \emph{Advances in Neural
  Information Processing Systems}, 2017, pp. 5434--5442.

\bibitem{halbawi2017improving}
W.~Halbawi, N.~Azizan, F.~Salehi, and B.~Hassibi, ``Improving distributed
  gradient descent using reed-solomon codes,'' in \emph{IEEE International
  Symposium on Information Theory}, 2018, pp. 2027--2031.

\bibitem{DCG16}
S.~Dutta, V.~Cadambe, and P.~Grover, ``Short-dot: Computing large linear
  transforms distributedly using coded short dot products,'' in
  \emph{Conference on Neural Information Processing Systems}, 2016, pp.
  2092--2100.

\bibitem{KS18}
Y.~Keshtkarjahromi, Y.~Xing, and H.~Seferoglu, ``Dynamic heterogeneity-aware
  coded cooperative computation at the edge,'' in \emph{IEEE International
  Conference on Network Protocols}, 2018, pp. 23--33.

\bibitem{fahim2017optimal}
S.~Dutta, M.~Fahim, F.~Haddadpour, H.~Jeong, V.~Cadambe, and P.~Grover, ``On
  the optimal recovery threshold of coded matrix multiplication,'' \emph{IEEE
  Transactions on Information Theory}, vol.~66, no.~1, pp. 278--301, 2019.

\bibitem{BPR17}
R.~Bitar, P.~Parag, and S.~El~Rouayheb, ``Minimizing latency for secure coded
  computing using secret sharing via staircase codes,'' \emph{IEEE Transactions
  on Communications}, vol.~68, no.~8, pp. 4609--4619, 2020.

\bibitem{Amiri2019}
M.~M. Amiri and D.~Gündüz, ``Computation scheduling for distributed machine
  learning with straggling workers,'' \emph{IEEE Transactions on Signal
  Processing}, vol.~67, no.~24, p. 6270–6284, Dec 2019.

\bibitem{Ozfatura2020}
E.~Ozfatura, S.~Ulukus, and D.~Gündüz, ``Straggler-aware distributed
  learning: Communication–computation latency trade-off,'' \emph{Entropy},
  vol.~22, no.~5, 2020.

\bibitem{kadhe2020communication}
S.~Kadhe, O.~O. Koyluoglu, and K.~Ramchandran, ``Communication-efficient
  gradient coding for straggler mitigation in distributed learning,'' in
  \emph{IEEE International Symposium on Information Theory}, 2020, pp.
  2634--2639.

\bibitem{charles2017approximate}
Z.~Charles, D.~Papailiopoulos, and J.~Ellenberg, ``Approximate gradient coding
  via sparse random graphs,'' \emph{arXiv preprint arXiv:1711.06771}, 2017.

\bibitem{MRM18}
R.~K. Maity, A.~S. Rawat, and A.~Mazumdar, ``Robust gradient descent via moment
  encoding and {LDPC} codes,'' in \emph{IEEE International Symposium on
  Information Theory}, 2019, pp. 2734--2738.

\bibitem{bitar2020stochastic}
R.~Bitar, M.~Wootters, and S.~El~Rouayheb, ``Stochastic gradient coding for
  straggler mitigation in distributed learning,'' \emph{IEEE Journal on
  Selected Areas in Information Theory}, vol.~1, no.~1, pp. 277--291, 2020.

\bibitem{wang2019erasurehead}
H.~Wang, Z.~Charles, and D.~Papailiopoulos, ``Erasurehead: Distributed gradient
  descent without delays using approximate gradient coding,'' \emph{arXiv
  preprint arXiv:1901.09671}, 2019.

\bibitem{wang2019fundamental}
S.~Wang, J.~Liu, and N.~Shroff, ``Fundamental limits of approximate gradient
  coding,'' \emph{ACM on Measurement and Analysis of Computing Systems},
  vol.~3, no.~3, pp. 1--22, 2019.

\bibitem{glasgow2021approximate}
M.~Glasgow and M.~Wootters, ``Approximate gradient coding with optimal
  decoding,'' \emph{IEEE Journal on Selected Areas in Information Theory},
  vol.~2, no.~3, pp. 855--866, 2021.

\bibitem{Ozfatura2021}
E.~Ozfatura, S.~Ulukus, and D.~Gündüz, ``Coded distributed computing with
  partial recovery,'' \emph{IEEE Transactions on Information Theory}, vol.~68,
  no.~3, pp. 1945--1959, 2022.

\bibitem{sarmasarkar2021gradient}
S.~Sarmasarkar, V.~Lalitha, and N.~Karamchandani, ``On gradient coding with
  partial recovery,'' in \emph{IEEE International Symposium on Information
  Theory}, 2021, pp. 2274--2279.

\bibitem{CPMBJ16}
J.~Chen, X.~Pan, R.~Monga, S.~Bengio, and R.~Jozefowicz, ``Revisiting
  distributed synchronous {SGD},'' \emph{arXiv preprint arXiv:1604.00981},
  2016.

\bibitem{RM51}
H.~Robbins and S.~Monro, ``A stochastic approximation method,'' \emph{The
  Annals of Mathematical Statistics}, vol.~22, no.~3, pp. 400--407, 1951.

\bibitem{Murata}
N.~Murata, ``A statistical study of on-line learning,'' \emph{Online Learning
  and Neural Networks. Cambridge University Press, Cambridge, UK}, pp. 63--92,
  1998.

\bibitem{pesme2020convergence}
S.~Pesme, A.~Dieuleveut, and N.~Flammarion, ``On convergence-diagnostic based
  step sizes for stochastic gradient descent,'' in \emph{International
  Conference on Machine Learning}, 2020, pp. 7641--7651.

\bibitem{Ge2019}
R.~Ge, S.~M. Kakade, R.~Kidambi, and P.~Netrapalli, ``The step decay schedule:
  A near optimal, geometrically decaying learning rate procedure for least
  squares,'' in \emph{Advances in Neural Information Processing Systems},
  vol.~32, 2019.

\bibitem{Smith2018}
S.~L. Smith, P.-J. Kindermans, C.~Ying, and Q.~V. Le, ``Don't decay the
  learning rate, increase the batch size,'' 2018, arXiv preprint
  arXiv:1711.00489.

\bibitem{Dutta2018}
S.~Dutta, G.~Joshi, S.~Ghosh, P.~Dube, and P.~Nagpurkar, ``Slow and stale
  gradients can win the race: Error-runtime trade-offs in distributed {SGD},''
  in \emph{International Conference on Artificial Intelligence and Statistics},
  A.~Storkey and F.~Perez-Cruz, Eds., vol.~84, Apr 2018, pp. 803--812.

\bibitem{KasHanna2020}
S.~Kas~Hanna, R.~Bitar, P.~Parag, V.~Dasari, and S.~El~Rouayheb, ``Adaptive
  distributed stochastic gradient descent for minimizing delay in the presence
  of stragglers,'' \emph{IEEE International Conference on Acoustics, Speech and
  Signal Processing}, May 2020.

\bibitem{hanna2022adaptive}
S.~Kas~Hanna, R.~Bitar, P.~Parag, V.~Dasari, and S.~E. Rouayheb, ``Adaptive
  stochastic gradient descent for fast and communication-efficient distributed
  learning,'' \emph{arXiv preprint arXiv:2208.03134}, 2022.

\bibitem{Pflug1990}
G.~C. Pflug, ``Non-asymptotic confidence bounds for stochastic approximation
  algorithms with constant step size,'' \emph{Monatshefte f{\"u}r Mathematik},
  vol. 110, no.~3, pp. 297--314, 1990.

\bibitem{Chee2018}
J.~Chee and P.~Toulis, ``Convergence diagnostics for stochastic gradient
  descent with constant learning rate,'' in \emph{International Conference on
  Artificial Intelligence and Statistics}, vol.~84, Apr 2018, pp. 1476--1485.

\bibitem{schlegel2022privacy}
R.~Schlegel, S.~Kumar, E.~Rosnes, and A.~Graell~i Amat, ``Privacy-preserving
  coded mobile edge computing for low-latency distributed inference,''
  \emph{IEEE Journal on Selected Areas in Communications}, vol.~40, no.~3, pp.
  788--799, 2022.

\bibitem{dhakal2019coded}
S.~Dhakal, S.~Prakash, Y.~Yona, S.~Talwar, and N.~Himayat, ``Coded computing
  for distributed machine learning in wireless edge network,'' in
  \emph{Vehicular Technology Conference}, 2019, pp. 1--6.

\bibitem{Bottou2018}
L.~Bottou, F.~E. Curtis, and J.~Nocedal, ``Optimization methods for large-scale
  machine learning,'' \emph{SIAM Review}, vol.~60, no.~2, pp. 223--311, 2018.

\bibitem{Liang2014}
G.~Liang and U.~C. Kozat, ``Tofec: Achieving optimal throughput-delay trade-off
  of cloud storage using erasure codes,'' in \emph{IEEE Conference on Computer
  Communications}, 2014, pp. 826--834.

\bibitem{Lee2016}
K.~Lee, M.~Lam, R.~Pedarsani, D.~Papailiopoulos, and K.~Ramchandran, ``Speeding
  up distributed machine learning using codes,'' in \emph{IEEE International
  Symposium on Information Theory}, 2016, pp. 1143--1147.

\bibitem{Renyi1953-vo}
A.~R{\'e}nyi, ``\BIBforeignlanguage{en}{On the theory of order statistics},''
  \emph{\BIBforeignlanguage{en}{Acta Math. Hungar.}}, vol.~4, no. 3-4, pp.
  191--231, Sep. 1953.

\end{thebibliography}
